\documentclass{article}
\usepackage[utf8]{inputenc}

\usepackage{amsmath,amssymb,amsfonts,amsthm,float}

\usepackage[nottoc]{tocbibind}
\usepackage[main=british,british]{babel}
\usepackage[backend=biber, style=numeric,]{biblatex}
\usepackage[compact]{titlesec}

\titleformat{\chapter}[display]
{\normalfont\huge\bfseries}{\chaptertitlename\ \thechapter}{20pt}{\Huge}

\titlespacing*{\chapter}{0pt}{-40pt}{30pt}
\DeclareRobustCommand{\bbinom}{\genfrac{[}{]}{0pt}{}}

\theoremstyle{plain}
\newtheorem{thm}{Theorem}[section]
\newtheorem{lem}[thm]{Lemma}
\newtheorem{prop}[thm]{Proposition}
\newtheorem{cor}[thm]{Corollary}

\theoremstyle{definition}
\newtheorem{defn}[thm]{Definition}

\newtheorem{exmp}[thm]{Example}

\theoremstyle{remark}

\newtheorem*{note}{Note}

\usepackage{chngcntr}
\counterwithin{table}{section}
\counterwithin{figure}{section}
\counterwithin{equation}{section}

\usepackage{mathrsfs}

\usepackage{xcolor}

\usepackage{aligned-overset}
\usepackage{geometry}
\title{The MacWilliams Identity for the Hermitian Rank Metric}
\author{Izzy Friedlander\footnote{Department of Computer Science, Durham University, UK. \texttt{isobel.s.friedlander@durham.ac.uk}}
}
\date{\today}

\begin{document}

\maketitle

\thispagestyle{plain}

\begin{abstract}
Error-correcting codes have an important role in data storage and transmission and in cryptography, particularly in the post-quantum era. Hermitian matrices over finite fields and equipped with the rank metric have the potential to offer enhanced security with greater efficiency in encryption and decryption. One crucial tool for evaluating the error-correcting capabilities of a code is its weight distribution and the MacWilliams Theorem has long been used to identify this structure of new codes from their known duals. Earlier papers have developed the MacWilliams Theorem for certain classes of matrices in the form of a functional transformation, developed using $q$-algebra, character theory  and Generalised Krawtchouk polynomials, which is easy to apply and also allows for moments of the weight distribution to be found. In this paper, recent work by Schmidt on the properties of codes based on Hermitian matrices such as bounds on their size and the eigenvalues of their association scheme is extended by introducing a negative-$q$ algebra to establish a MacWilliams Theorem in this form together with some of its associated moments. The similarities in this approach and in the paper for the Skew-Rank metric by Friedlander et al. \cite{friedlander2023macwilliams} have been emphasised to facilitate future generalisation to any translation scheme. 
\end{abstract}


\section{Introduction}

Error-correcting codes have been used for transmitting and storing data over potentially unreliable media for many years. More recently they have been adapted to encrypt data for secure transmission and storage, most notably using the McEliece cryptosystem, originally designed using random permutations of binary Goppa codes with added random errors \cite{McElicseCrypto} With the development of quantum computers there is pressure to identify increasingly secure encryption algorithms and error-correcting codes continue to have that potential. One such algorithm, “Classic McEliece”, is a finalist in the Post Quantum Cryptography (PQC) competition run by the National Institute of Standards and Technology (NIST) but there are performance concerns for some applications because of its large public key size and slow key generation  \cite{NISTLIST}.

There is the potential to improve the performance of McEliece type cryptosystems by using different underlying codes. In particular codes based on a different distance metric, the rank metric, can have significantly smaller public keys \cite{PierreLoidreau}\cite{CodebasedencrpytionLau} while maintaining their security and error correction capability. 

Codes based on the rank metric were initially explored by Delsarte \cite{DelsarteBlinear}\cite{DelsarteAlternating} and Gabidulin \cite{GabidulinTheory} and many results have been derived for codes based on general matrices \cite{DelsarteBlinear}, symmetric matrices \cite{KaiApplications}\cite{SymmetricKai} and skew-symmetric or alternating matrices \cite{DelsarteAlternating} over finite fields of varying sizes. Two critical parameters used to measure the effectiveness of a code are its size (the number of codewords) which impacts efficiency, and its weight distribution which records the distances between codewords and impacts on its error correction strength. 
 
Hermitian matrices over a finite field are another possible underlying structure for codes based on the rank metric. Carlitz and Hodges \cite{CarlitzCountingHermitian} in 1955 identified many properties of such matrices,  and Stanton \cite{StantonKrawtchoukPolys} used hypergeometric series to find the eigenvalues of their association scheme in 1979. More recently there have been studies of the rank properties of these matrices \cite{Sheeky}\cite{Gow}. It wasn’t until much more recently that knowledge of codes based on Hermitian matrices was extended significantly by Schmidt \cite{KaiHermitian}. Schmidt uses results from the theory of association schemes \cite{delsarte1973algebraic}\cite{InfoTheoryDelsarte},and in particular the eigenvalues of the schemes, to prove many upper and lower bounds on the size of a code with a given minimum distance between codewords and to identify properties of its weight distribution. 

Delsarte \cite{delsarte1973algebraic} applied the theory of association schemes to the well known MacWilliams theorem \cite{TheoryofError} relating the weight distribution of a code to that of its dual code. Gadouleau and Yan \cite{gadouleau2008macwilliams} also derived the MacWilliams identity in the form of a functional transformation using character theory, $q$-algebra and the Hadamard Transform \cite{TheoryofError}. Friedlander, Gadouleau and Bouganis \cite{friedlander2023macwilliams} extended that functional transformation structure to codes based on skew-symmetric matrices using Generalised Krawtchouk polynomials and $q$-algebra.  By defining appropriate $q$-derivatives, the MacWilliams Theorem was used to relate the moments of the weight distributions of a code and its dual as well. 

In Hermitian Rank Distance codes by Kai-Uwe Schmidt \cite{KaiHermitian} he develops bounds on code sizes and proves the conditions under which the weight distribution is dependent only on it's parameters. In this paper, we draw heavily on the results of Schmidt, in particular the bounds and the eigenvalues of the association scheme to build a relevant negative-$q$-algebra. In this way the MacWilliams Theorem has been proven in the form once again of a new functional transformation. By introducing negative-$q$ derivatives, the relationship between moments of the weight distribution of a code and its dual have also been established.

This paper applies the same methodology to the association scheme of Hermitian matrices based on the rank metric. Once again an additional $q$-analog MacWilliams Identity is derived using the eigenvalues of the association scheme in place of the generalised Krawtchouk polynomials. All the necessary $q$-algebra was adapted for this application.

The rest of this paper is structured as follows: In Section \ref{section:Preliminaries} the necessary definitions and properties are introduced and some important identities are derived. Section \ref{section:negative-q-product} defines the negative-$q$-product, negative-$q$-power and negative-$q$-transform for homogeneous polynomials. Once again in particular, the powers of two specific key polynomials are found and related to the weight enumerators of Hermitian matrices of any size. In Section \ref{section:MacWilliamsHermitian} a negative-$q$ form of the generalised Krawtchouk polynomials is established and is used to prove a $q$-analog of the MacWilliams identity for the rank metric as a functional transform for Hermitian matrices. Section \ref{section:derivatives} introduces two negative-$q$-derivatives for real valued functions of a variable and derives some results for homogeneous polynomials including the two key polynomials explored in Section \ref{section:negative-q-product}. The derivatives are then used in Section \ref{section:moments} to identify moments of the rank distribution for linear codes based on Herimitian matrices. 

The results presented in this paper are included in \cite{IzzyThesis}, and they open clearly the possibility of obtaining similar results for multiple association schemes. It is now clear that this may be extended to more general schemes such as translation association schemes. The crucial question here is whether one can define the analogue of the $q$-product in a general setting such that the MacWilliams identity can be stated in a functional form, as the one in \cite{gadouleau2008macwilliams}, \cite{friedlander2023macwilliams} and the one obtained here.
\section{Preliminaries}\label{section:Preliminaries}

\subsection{Hermitian Matrices}
We begin with some assumptions and definitions relating to Hermitian Matrices over a finite field, $\mathbb{F}_{q^2}$ where $q$ is a prime power. We follow the convention that the empty product is taken to be $1$ and the empty sum is taken to be $0$. We also define $\sigma_i=\frac{i(i-1)}{2}$ for $i\geq0.$ 

We write the conjugate of $x\in\mathbb{F}_{q^2}$ as $\overline{x}= x^q$, Similarly for a $t \times t$ matrix over $\mathbb{F}_{q^2}$, we write $\boldsymbol{H}^\dag$ for the conjugate transpose matrix of $\boldsymbol{H}$.
\begin{defn}
    Let $\boldsymbol{H}$ be a $t \times t$ matrix over $\mathbb{F}_{q^2}$. Then $\boldsymbol{H}=(h_{ij})$ is called a \textbf{\textit{Hermitian}} matrix if $\boldsymbol{H}=\boldsymbol{H}^\dag$. The set of these Hermitian matrices is denoted $\mathscr{H}_{q,t}$.
\end{defn}

Each Hermitian matrix, $\boldsymbol{H}$, can be associated with a corresponding Hermitian Form, which is a bilinear form
\begin{equation}
    \boldsymbol{H}: ~ V \times V \rightarrow \mathbb{F}_{q^2}
\end{equation}
where $V$ is a $t$ dimensional space over $\mathbb{F}_{q^2}$ with $\left\{\boldsymbol{e}_1,\boldsymbol{e}_2,\ldots,\boldsymbol{e}_t\right\}$ and 
\begin{equation}
    \boldsymbol{H}\left(\boldsymbol{e}_i,\boldsymbol{e}_j \right) = h_{ij}.
\end{equation}
The set of these bilinear forms is denoted $\mathbb{B}(t,q)$. There is a one to one correspondence between $\mathscr{H}_{q,t}$ and $\mathbb{B}(t,q)$.

\begin{thm}\label{theorem:vectorspaceHermitian}
    $\mathscr{H}_{q,t}$ is a $t^2$-dimensional vector space over $\mathbb{F}_{q}$.
\end{thm}

\begin{proof}
    The proof of Theorem \ref{theorem:vectorspaceHermitian} is trivial and hence omitted.
\end{proof}







\begin{defn}\label{defn:rankofhermitianmatrix}
For all $\boldsymbol{H},\boldsymbol{J}\in \mathscr{H}_{q,t}$, we define the \textit{\textbf{rank distance}} to be 
\begin{equation}
    d_{R}(\boldsymbol{H},\boldsymbol{J})=R(\boldsymbol{H}-\boldsymbol{J})
\end{equation}
where $R(\boldsymbol{H})$ is the column rank of $\boldsymbol{H}$. It is easily verified that $d_{R}$ is a metric over $\mathscr{H}_{q,t}$.
\end{defn}

\subsection{Codes based on Subspaces of Hermitian Matrices}

Any subspace of $\mathscr{H}_{q,t}$ can be considered as an $\mathbb{F}_{q}$-linear code, $\mathscr{C}$, with each matrix of rank $r$ in $\mathscr{C}$ representing a codeword of weight $r$ and with the distance metric being the rank metric defined in Defintion \ref{defn:rankofhermitianmatrix}.

The \textbf{\textit{minimum rank distance}} of such a code $\mathscr{C}$, denoted as $d_{R}(\mathscr{C})$, is simply the minimum rank distance over all possible pairs of distinct codewords in $\mathscr{C}$. When there is no ambiguity about $\mathscr{C}$, we denote the minimum rank distance as $d_{R}$.

The following bound is proven explicitly in \cite[Theorem 1]{KaiHermitian} using the relationship between the inner distributions and the eigenvalues of the association scheme (which happens to actually obtain one of the sets of moments of the association scheme). It then uses the subgroup properties of the code to show that the bound holds. So we have that the cardinality $|\mathscr{C}|$ of a code $\mathscr{C}$ over $\mathbb{F}_{q^2}$ based on $t \times t$ Hermitian matrices and minimum rank distance $d_{R}$ satisfies 
\begin{equation}\label{SinglebuondHermitian}
    |\mathscr{C}|\leq q^{t(t-d_{R}+1)}.
\end{equation}

In this paper, we call the bound in \eqref{SinglebuondHermitian} the Singleton Bound for codes with the Hermitian Rank Metric.
Codes that attain the Singleton bound are referred to as Maximal Codes or Maximum Hermitian Rank Distance (MHRD) codes. 

\begin{defn}\label{def:Hermitianweightenumerator}
For all $\boldsymbol{H}\in\mathscr{H}_{q,t}$ with rank weight $r$, the \textit{\textbf{rank weight function}} of $\boldsymbol{H}$ is defined as the homogeneous polynomial
\begin{equation}
f_{R}(\boldsymbol{H})=Y^{r}X^{t-r}.
\end{equation}
Let $\mathscr{C} \subseteq \mathscr{H}_{q,t}$ be a code. Suppose there are $c_i$ codewords in $\mathscr{C}$ with rank weight $i$ for $0\leq i\leq t$. Then the \textbf{\textit{rank weight enumerator}} of $\mathscr{C}$, denoted as $W_\mathscr{C}^{R}(X,Y)$ is defined to be
\begin{equation}\label{defn:skewrankweightenumerator}
    W_\mathscr{C}^{R}(X,Y) = \sum_{\boldsymbol{H}\in\mathscr{C}}f_{R}(\boldsymbol{H})=\sum_{i=0}^{t} c_iY^{i}X^{t-i}.
\end{equation}

The $(t+1)$-tuple, $\boldsymbol{c}=(c_0,\ldots,c_t)$ of coefficients of the rank weight enumerator, is called the \textbf{\textit{weight distribution}} of the code $\mathscr{C}$.
\end{defn}

\begin{exmp}
    An example of such a code with $q=2$ and $t=3$ is where $\mathscr{C}$ is the set of Hermitian matrices, $\boldsymbol{H}$ over $\mathbb{F}_{4}$ such that;
    \begin{equation}
        \mathscr{C} = \left\{~
        \begin{aligned}
            \begin{pmatrix}
                    0 & 0 & 0 \\
                    0 & 0 & 0 \\
                    0 & 0 & 0 
                \end{pmatrix},
            \begin{pmatrix}
                    1 & \alpha & 0 \\
                    1+\alpha & 0 & 0 \\
                    0 & 0 & 0 
                \end{pmatrix}, &
            \begin{pmatrix}
                    0 & 0 & \alpha \\
                    0 & 1 & 0 \\
                    1+\alpha & 0 & 0 
                \end{pmatrix},
            \begin{pmatrix}
                    0 & 0 & 0 \\
                    0 & 0 & 1+\alpha \\
                    0 & \alpha & 1 
                \end{pmatrix}\\
            \begin{pmatrix}
                    1 & \alpha & \alpha \\
                    1+\alpha & 1 & 0 \\
                    1+\alpha & 0 & 0 
                \end{pmatrix},
            \begin{pmatrix}
                    1 & \alpha & 0 \\
                    1+\alpha & 0 & 1+\alpha \\
                    0 & \alpha & 1 
                \end{pmatrix}, &
            \begin{pmatrix}
                    0 & 0 & \alpha \\
                    0 & 1 & 1+\alpha \\
                    1+\alpha & \alpha & 1 
                \end{pmatrix},
            \begin{pmatrix}
                    1 & \alpha & \alpha \\
                    1+\alpha & 1 & 1+\alpha \\
                    1+\alpha & \alpha & 1 
                \end{pmatrix}
        \end{aligned}
    ~\right\}
    \end{equation}
    with multiplication table below.
    \begin{table}[H]
        \centering
        \begin{tabular}{||c||cccc||}
        \hline
            & 0 & 1 & $\alpha$ & $1+\alpha$\\
            \hline\hline 
            0 & 0 & 0 & 0 & 0 \\
            1 & 0 & 1 & $\alpha$ & $1+\alpha$ \\
            $\alpha$ & 0 & $\alpha$ & $1+\alpha$ & 1 \\
            $1+\alpha$ & 0 & $1+\alpha$ & 1 & $\alpha$\\
            \hline
        \end{tabular}
        \label{tab:examplemultiplicationtable}
    \end{table}
    Thus there are $8$ matrices (codewords) in this code, 1 of Hermitian rank 0, 3 of Hermitian rank 2 and 4 of Hermitian rank 3. Thus its Hermitian rank weight enumerator is $X^3 + 3Y^2X + 4Y^3$.
\end{exmp}

\subsection{Counting the number of Hermitian Matrices of a given size}
Multiple ways of describing the number of Hermitian matrices exist. The following is (for the purpose of this paper) in the best format and has been adapted from \cite[Theorem 3, p398]{CarlitzCountingHermitian}.

\begin{thm}\label{thm:hermitiancountingMatrices}
    The number of Hermitian matrices of order $t$ and Hermitian rank weight $h$ is given by
    \begin{equation}
        \xi_{t,h}=q^{\sigma_h} \times \frac{\displaystyle\prod_{i=0}^{h-1}q^{2t-2i}-1}{\displaystyle\prod_{i=1}^{h}q^{i}-(-1)^{i}} = (-1)^{h}(-q)^{\sigma_h}\times \frac{\displaystyle\prod_{i=0}^{h-1}(-q)^{2t-2i}-1}{\displaystyle\prod_{i=1}^{h}(-q)^i-1}.
    \end{equation}
\end{thm}

We also then note the Hermitian rank weight enumerator of $\mathscr{H}_{q,t}$ is
\begin{equation}
    \Omega_t= \sum_{i=0}^t \xi_{t,i}Y^{i}X^{t-i}.
\end{equation}

\begin{exmp}
    For $t=3$ and $q=2$ the rank weight enumerator of $\mathscr{H}_{2,3}$ is
    \begin{align}
        X^3 + 21YX^2 + 210Y^2X + 280Y^3.
    \end{align}
\end{exmp}


\subsection{Inner product of two Hermitian Matrices}

We define an \textbf{inner product} on $\mathscr{H}_{q,t}$ by
\begin{equation}
    (\boldsymbol{H},\boldsymbol{J})\mapsto\langle \boldsymbol{H},\boldsymbol{J} \rangle =Tr\left(\boldsymbol{H}^\dag\boldsymbol{J}\right)
\end{equation}
where $Tr(\boldsymbol{H})$ means the trace of $\boldsymbol{H}$.

\begin{defn}
The \textbf{\textit{dual}} of a code, $\mathscr{C}$, denoted by $\mathscr{C}^\perp$ is defined as
\begin{equation}
    \mathscr{C}^\perp = \big\{ \boldsymbol{H}\in\mathscr{H}_{q,t}~ \vert
\left\langle \boldsymbol{H},\boldsymbol{J}\right\rangle =0~\forall~ \boldsymbol{J} \in \mathscr{C}\big\}.\end{equation}
\end{defn}




\subsection{Negative-$q$ Gaussian Coefficients and other useful identities}

In establishing the results later in this paper we have used some identities to simplify the notation and algebra.


\begin{defn}\label{definition:negativeqGaussianCoeff}
    For any real number $q\neq 1,~ b=-q,~k\in\mathbb{Z}^+$ and $x\in\mathbb{R}$ (usually an integer), we define the \textbf{\textit{Negative-$q$ Gaussian Coefficient}} \cite{Negativeqbinomial}, {$\begin{bmatrix} x \\ k \end{bmatrix}$}, to be
    \begin{equation}
        \begin{bmatrix} x \\ k \end{bmatrix} = \prod_{i=0}^{k-1} \frac{b^{x}-b^i}{b^{k}-b^{i}}
    \end{equation}
    with
    \begin{equation}
        \begin{bmatrix} x \\ 0 \end{bmatrix} = 1.
    \end{equation}
\end{defn}

\begin{lem}
    The Negative-$q$ Gaussian Coefficient above can be written as
    \begin{equation}
        \begin{bmatrix} x \\ k \end{bmatrix} = \prod_{i=1}^k \frac{b^{x-i+1}-1}{b^{i}-1}.
    \end{equation}
\end{lem}

\begin{proof}
    We have
    \begin{align}
        \begin{bmatrix} x \\ k \end{bmatrix}   
            & = \prod_{i=0}^{k-1} \frac{b^{x}-b^i}{b^{k}-b^{i}}
            \\
            & = \prod_{i=0}^{k-1} \frac{b^{x-i}-1}{b^{k-i}-1}
            \\
            & = \prod_{i=1}^{k} \frac{b^{x-i+1}-1}{b^{i}-1}.
    \end{align}
\end{proof}

Here are some identities relating the negative-$q$ Gaussian coefficients that are useful from \cite[Section 2.2]{DelsarteAlternating}.
\begin{align}
    \left[\begin{matrix} x \\ k\end{matrix}\right] 
        & = \left[\begin{matrix} x \\ x-k\end{matrix}\right]\label{equation:gaussianxx-k}\\
    \left[\begin{matrix} x \\ i\end{matrix}\right]\left[\begin{matrix} x-i \\ k\end{matrix}\right] 
        & = \left[\begin{matrix} x \\ k\end{matrix}\right]\left[\begin{matrix} x-k \\ i\end{matrix}\right]\label{equation:gaussianswapplaces}\\
    \prod_{i=0}^{x-1}\left(y-b^{i}\right)  
        & = \sum_{k=0}^x (-1)^{x-k}b^{\binom{x-k}{2}}\left[\begin{matrix} x \\ k\end{matrix}\right]y^k\\
    \sum_{k=0}^x\left[\begin{matrix} x \\ k\end{matrix}\right]\prod_{i=0}^{k-1}\left(y-b^{i}\right)  
        & = y^x\label{equation:producttosumgauss}\\
    \sum_{k=i}^j (-1)^{k-i}b^{\sigma_{k-i}}\left[\begin{matrix} k \\ i\end{matrix}\right]\left[\begin{matrix} j \\ k\end{matrix}\right] 
        & = \delta_{ij}.\label{equation:deltaijbs}
\end{align}

The following identities are each used in the rest of this paper but can be shown trivially to be equal.

\begin{align}
    \begin{bmatrix} x \\ k \end{bmatrix} 
        & = \begin{bmatrix} x-1 \\ k \end{bmatrix} + b^{x-k}\begin{bmatrix} x-1 \\ k-1 \end{bmatrix} \label{equation:Stepdown1}
        \\ 
        & = \begin{bmatrix} x-1 \\ k-1 \end{bmatrix} + b^{k}\begin{bmatrix} x-1 \\ k \end{bmatrix} \label{equation:gaussiancoeffsx-1k-1}
        \\
        & = \frac{b^{x-k+1}-1}{b^{k}-1}\begin{bmatrix} x \\ k-1 \end{bmatrix}
        \label{equation:gaussianfracxk-1}
        \\
        & = \frac{b^{x}-1}{b^{x-k}-1}\begin{bmatrix} x-1 \\ k \end{bmatrix} \label{equation:gaussianfracx-1k}
        \\
                & = \frac{b^{x}-1}{b^{k}-1} \begin{bmatrix} x-1 \\ k-1 \end{bmatrix}. \label{equation:beta1stepdown}
        \end{align}

\begin{defn}
    We define the \textbf{\textit{Negative-$q$ Gamma function}} for $x\in\mathbb{R},~k\in\mathbb{Z}$ to be
    \begin{align}
        \gamma(x,k) 
            & = (-1)^{k}\prod_{i=0}^{k-1}\left( b^x + b^i\right)
            \\
            & = \prod_{i=0}^{k-1}\left( -b^{x} - b^{i}\right)
    \end{align}
\end{defn}

\begin{thm}\label{theorem:quicknumberHermitian}
    The number of Hermitian matrices of order $t$ and rank $r$ can be rewritten as
    \begin{equation}
        \xi_{t,r} = \begin{bmatrix} t \\ r \end{bmatrix}\gamma(t,r).
    \end{equation}
\end{thm}

\begin{proof}
    We have
    \begin{align}
        \begin{bmatrix} t \\ r \end{bmatrix}\gamma(t,r)
            & = (-1)^r \displaystyle\prod_{i=1}^{r} \frac{b^{t-i+1}-1}{b^{i}-1}\displaystyle\prod_{i=0}^{k-1} \left(b^{t}+b^{i}\right)
            \\
            & = (-1)^r \frac{\displaystyle\prod_{i=1}^r b^{t-i+1}-1}{\displaystyle\prod_{i=1}^{k}b^{i}-1}\displaystyle\prod_{i=0}^{k-1}\left(b^{t}+b^{i}\right)
            \\
            & = (-1)^r \frac{\displaystyle\prod_{i=0}^{r-1}b^{t-i}-1}{\displaystyle\prod_{i=1}^{r}(-1)^{i}\left(q^i-(-1)^i\right)}\displaystyle\prod_{i=0}^{r-1}\left(b^{t-i}+1\right)b^{i}
            \\
            & = (-1)^r \frac{\displaystyle\prod_{i=0}^{r-1}b^{i}\left(b^{2t-2i}-1\right)}{\displaystyle\prod_{i=1}^{r}(-1)^{i}\left(q^i-(-1)^i\right)}
            \\
            & = (-1)^rb^{\sigma_r} \frac{\displaystyle\prod_{i=0}^{r-1}b^{2t-2i}-1}{\displaystyle\prod_{i=1}^{r}(-1)^{i}\left(q^i-(-1)^i\right)}
            \\
            & = \frac{(-1)^r(-1)^{\sigma_r}q^{\sigma_r}}{(-1)^{\sigma_{r+1}}} \frac{\displaystyle\prod_{i=0}^{r-1}q^{2t-2i}-1}{\displaystyle\prod_{i=1}^{r}q^i-(-1)^i}
            \\
            & = q^{\sigma_r}\times \frac{\displaystyle\prod_{i=0}^{r-1}q^{2t-2i}-1}{\displaystyle\prod_{i=1}^{r}q^i-(-1)^i}
            \\
            & = \xi_{t,r}.
    \end{align}
\end{proof}

\begin{lem}\label{lemma:Gammaidentites}
We have the following identities for the negative-$q$ Gamma function:
\begin{align}
    \gamma(x,k)
        & = b^{\sigma_k}\displaystyle\prod_{i=0}^{k-1}\left(-b^{x-i}-1\right)
    \\
    \gamma(x,k) 
        & = b^{k-1} \big(-b^{x}-1\big) \gamma(x-1,k-1) \label{equation:gammastepdown}
    \\
    \gamma(x,k+1) 
        & = \left(-b^x-b^{k}\right)\gamma(x,k).\label{equation:gammastepdownsecond}
\end{align}


\end{lem}

\begin{proof}~\\
$(1)$
\begin{align}
   \gamma(x,k) 
        & = \prod_{i=0}^{k-1}\left(-b^x - b^{i}\right)
        \\
        & = \left(\prod_{i=0}^{k-1}b^{i}\right)\prod_{i=0}^{k-1}\left(-b^{x-i}-1\right)
        \\
        & = b^{\sigma_k}\prod_{i=0}^{k-1}\left(-b^{x-i}-1\right).
\end{align}


$(2)$
\begin{align}
        \gamma(x,k) 
                & =  \prod_{i=0}^{k-1} \left(-b^{x} - b^{i}\right)
                \\
                & =  \left(-b^{x}-1\right) \prod_{i=1}^{k-1}\left(-b^{x}-b^{i}\right)
                \\
                & =  \left(-b^{x}-1\right) \prod_{i=1}^{k-1} b\left(-b^{x-1}-b^{i-1}\right)
                \\
                & =  \left(-b^{x}-1\right) b^{k-1} \prod_{i=0}^{k-2} \left(-b^{x-1}-b^{i}\right)
                \\
                & = b^{k-1}\left(-b^{x}-1\right)\gamma(x-1,k-1).
\end{align}
$(3)$
\begin{align}
    \gamma(x,k+1) 
        & = \prod_{i=0}^k\left(-b^x-b^{i}\right)
        \\
        & = \left(-b^x-b^{k}\right)\prod_{i=0}^{k-1}\left(-b^x-b^{i}\right)
        \\
        & = \left(-b^x-b^{k}\right)\gamma(x,k).
\end{align}
\end{proof}

\begin{defn}\label{defn:negativebetafunction}
We also define a \textbf{\textit{Negative-$q$ Beta function}} for $x\in\mathbb{R}$, $k\in\mathbb{Z}^+$ as 
\begin{equation}\label{equation:betafunction}
    \beta(x,k) =
             \displaystyle\prod_{i=0}^{k-1}\begin{bmatrix}x-i\\1\end{bmatrix}.
\end{equation}
These are closely related to Negative-$q$ Gaussian Coefficients. 
\end{defn}

\begin{lem}\label{lemma:betabmanipulation}
We have for all $x\in\mathbb{R}$, $k\in\mathbb{Z}^+$,
\begin{equation}\label{equation:betabstartdifferent}
\beta(x,k) = \begin{bmatrix}x\\k\end{bmatrix}\beta(k,k)
\end{equation}
and
\begin{equation}\label{equation:betabstartsame}
\beta(x,x) = \begin{bmatrix}x\\k\end{bmatrix}\beta(k,k)\beta(x-k,x-k).
\end{equation}
\end{lem}

\begin{proof}
We have
\begin{align}
    \beta(x,k) = \prod_{i=0}^{k-1}\begin{bmatrix}x-i\\1\end{bmatrix} 
        & = \prod_{i=0}^{k-1}\frac{b^{(x-i)}-1}{b-1}
        \\
        & = \prod_{i=0}^{k-1}\frac{\left(b^{x-i}-1\right)\left(b^{k-i}-1\right)}{\left(b^{k-i}-1\right)(b-1)}
        \\
        & = \prod_{i=0}^{k-1}\frac{b^{x}-b^{i}}{b^{k}-b^{i}}\prod_{i=0}^{k-1}\frac{b^{k-i}-1}{b-1}
        \\
        & = {\begin{bmatrix}x\\k\end{bmatrix}}\beta(k,k)
\end{align}
as required. Now we have
\begin{align}
    {\begin{bmatrix}x\\k\end{bmatrix}}\beta(k,k)\beta(x-k,x-k) 
        & = \prod_{i=0}^{k-1}\frac{b^{x}-b^{i}}{b^{k}-b^{i}}\prod_{r=0}^{k-1}\frac{b^{k-r}-1}{b-1}\prod_{s=0}^{x-k-1}\frac{b^{x-k-s}-1}{b-1}
        \\
    %
        & = \prod_{i=0}^{x-1}\frac{b^{x-i}-1}{b-1}
        \\
        & = \beta(x,x)
\end{align}
as required.
\end{proof}

\begin{lem}\label{lemma:betaproperties}
    We have the following identity for the negative-q-beta function. For all $x\in\mathbb{R}$, $k\in\mathbb{Z}^+$ we have
    \begin{equation}\label{equation:betaproperties}
        \beta(x,k)\beta(x-k,1)=\beta(x,k+1).
    \end{equation}
\end{lem}

\begin{proof}
    \begin{align}
        \beta(x,k)\beta(x-k,1) 
            & = \begin{bmatrix} x-k \\ 1 \end{bmatrix} \prod_{i=0}^{k-1}\begin{bmatrix} x-i \\ 1 \end{bmatrix}
            \\
            & = \beta(x,k+1).
    \end{align}
\end{proof}

\section{The Negative-$q$-Product and Negative-$q$-Transform}\label{section:negative-q-product}

The weight enumerators of any linear code $\mathscr{C}\subseteq \mathscr{H}_{q,t}$ are homogeneous polynomials. We introduce an operation, the Negative-$q$-Product, on homogeneous polynomials that will help express the relation between the weight enumerator of a code and that of it's dual. 

\subsection{The Negative-$q$-product, Negative-$q$-power and the Negative-$q$-transform}

\begin{defn}\label{q-proddefn}
    Let
    \begin{align}
        a(X,Y;\lambda) 
            & = \sum_{i=0}^{r} a_i(\lambda) Y^{i} X^{r-i}
            \\
        b(X,Y;\lambda) 
            & = \sum_{i=0}^{s} b_i(\lambda) Y^{i} X^{s-i}
    \end{align}
    be two homogeneous polynomials in $X$ and $Y$, of degrees $r$ and $s$ respectively, and coefficients $a_i(\lambda)$ and $b_i(\lambda)$ respectively, which are real functions of $\lambda$ and are $0$ unless otherwise specified, for example $b_i(\lambda)=0$ if $i\notin\left\{0,1,\ldots,s\right\}$. The \textbf{\textit{negative-$q$-product}}, $\ast$, of $a(X,Y;\lambda)$ and $b(X,Y;\lambda)$, is defined as
    \begin{align}
        c(X,Y;\lambda) 
            & = a(X,Y;\lambda) \ast b(X,Y;\lambda)\label{equation:qproduct}
            \\
            & = \sum_{u=0}^{r+s} c_u(\lambda) Y^{u} X^{r+s-u} 
    \end{align} 
    with
    \begin{equation}
        c_u(\lambda) = \sum_{i=0}^{u} (-q)^{is} a_i(\lambda)b_{u-i}(\lambda -i).
    \end{equation}
\end{defn}

\begin{defn}
    As in \cite{gadouleau2008macwilliams} the \textbf{\textit{negative-$q$-power}} is defined by 
    \begin{equation}
        \begin{cases}
            a^{[0]}(X,Y;\lambda) = 1 \\
            a^{[1]}(X,Y;\lambda) = a(X,Y;\lambda) \\
            a^{[k]}(X,Y;\lambda) = a(X,Y;\lambda) \ast a^{[k-1]}(X,Y;\lambda) & \text{for } k \geq 2.
        \end{cases}
    \end{equation}
\end{defn}

\begin{defn}[{\cite[Definition 4]{gadouleau2008macwilliams}}]\label{defn:negativeqtransform}
    Let $a(X,Y;\lambda) = \sum_{u=0}^{r} a_{i}(\lambda) Y^{i} X^{r-i}$. We define the \textbf{\textit{negative-$q$-transform}} to be the homogeneous polynomial
    \begin{equation}
        \overline{a}(X,Y;\lambda) = \sum_{i=0}^{r} a_{i}(\lambda) Y^{[i]} \ast X^{[r-i]}.
    \end{equation}
\end{defn}

\subsection{Using the Negative-$q$-product to identify the Rank Weight Enumerator of Hermitian Matrices}\label{subsection:negativeqproduct}

In the theory that follows, relating the weight enumerator of a code to it's dual, then we consider the following polynomial. Let
\begin{equation}\label{}
    \mu(X,Y;\lambda) = X + \left(-b^{\lambda}-1\right)Y.
\end{equation}

\begin{thm}
    If $\mu(X,Y;\lambda)$ is as defined above, then
    \begin{equation} \label{equation:muformula}
        \mu^{[k]}(X,Y;\lambda) = \sum_{u=0}^{k} \mu_{u}(\lambda,k) Y^{u} X^{k-u} \quad \text{for } k \geq 1,
    \end{equation}
    where
    \begin{equation}
        \mu_{u}(\lambda,k) = \begin{bmatrix} k \\ u \end{bmatrix}\gamma(\lambda,u).
    \end{equation}
    Specifically, the weight enumerators for $\mathscr{H}_{q,t}$, the set of hermitian matrices of size $t \geq 1$, denoted by $\Omega_{t}$, is given by
    \begin{equation}
        \Omega_{t} = \mu^{[t]}(X,Y;t).
    \end{equation}
    In other words, the negative-$q$-powers of $\mu(X,Y;t)$ provide an explicit form for the weight enumerator of $\mathscr{H}_{q,t}$, the set of hermitian matrices of order $t$.
\end{thm}

\begin{proof}
    The proof follows the method of induction.
    Consider $k=1$, $\mu^{[1]}(X,Y;\lambda) = \mu(X,Y;\lambda)$ $= X + \left(-b^{\lambda}-1\right)Y$.
    \begin{align}
        \mu_{0}(\lambda,1) 
            & = 1 = \begin{bmatrix} 1 \\ 0 \end{bmatrix}\gamma(\lambda,0) 
            \\
        \mu_{1}(\lambda,1) 
            & = \left(-b^{\lambda}-1\right) = \begin{bmatrix} 1 \\ 1 \end{bmatrix} \gamma(\lambda,1)
    \end{align}
    so 
    \begin{equation}
        \Omega_{1} = \mu^{[1]}(X,Y;\lambda)
    \end{equation}
    and
    \begin{equation}
        \mu_{u}^{[1]}(\lambda,1) = \begin{bmatrix} 1 \\ u \end{bmatrix} \gamma(\lambda,u)
    \end{equation}
    as required for $k=1$. Now assume  the theorem is true for $k \geq 1$. Then 
    \begin{equation}
        \mu^{[k+1]}(X,Y;\lambda) = \mu \ast \mu^{[k]} = \sum_{i=0}^{k+1} f_{i}(\lambda) Y^{i} X^{k+1-i}
    \end{equation}
    where $f_{i}(\lambda) = \sum_{j=0}^{i} b^{jk} \mu_{j}(\lambda,1) \mu_{i-j}(\lambda-j,k)$ by definition of the negative-$q$-product. Then, 
    \begin{align}
        \mu_{0}(\lambda,1) \mu_{i}(\lambda,k) + b^{k} \mu_{1}(\lambda,1) \mu_{i-1}(\lambda-1,k)
            & = \begin{bmatrix} k \\ i \end{bmatrix} \gamma(\lambda,i) + b^{k}\left( -b^{\lambda}-1\right) \begin{bmatrix} k \\ i-1 \end{bmatrix} \gamma(\lambda-1,i-1) 
            \\
            & \text{by the inductive hypothesis,} 
            \\
            & = \frac{b^{k-i+1}-1}{b^{k+1}-1} \begin{bmatrix} k+1 \\ i \end{bmatrix} \gamma(\lambda,i) + b^{k} \frac{b^{i}-1}{b^{k+1}-1} b^{1-i} \begin{bmatrix} k+1 \\ i \end{bmatrix} \gamma(\lambda,i) 
            \\
            & = \gamma(\lambda,i) \begin{bmatrix} k+1 \\ i \end{bmatrix}\left( \frac{b^{k-i+1}-1+b^{k-i+1}\left(b^{i}-1\right)}{b^{k+1}-1} \right)
            \\
            & = \gamma(\lambda,i) \begin{bmatrix} k+1 \\ i \end{bmatrix}
    \end{align}
    so it is true for $k+1$. Therefore by induction the theorem is true.
\end{proof}

Now let $\nu(X,Y;\lambda) = X - Y$.

\begin{thm}\label{lemma:nulemma}
For all $k\geq 1$,
    \begin{equation}
        \nu^{[k]}(X,Y;\lambda) = \sum_{u=0}^{k} (-1)^{u}b^{\sigma_{u}} \begin{bmatrix} k \\ u \end{bmatrix} Y^{u} X^{k-u}
    \end{equation}
\end{thm}

\begin{proof}
    We perform induction on $k$. It is easily checked that the theorem holds for $k=1$.

    Now assume the theorem holds for $k\geq 1$.
    \begin{align}
        \nu^{[k+1]} 
            & = \nu \ast \nu^{[k+1]} 
            \\
            & = (X-Y) \ast \left(\sum_{u=0}^{k} (-1)^{u} b^{\sigma_u} \begin{bmatrix} k \\ u \end{bmatrix} Y^{u} X^{k-u} \right)
            \\
            & = \sum_{i=0}^{k+1} g_i(\lambda) Y^{i} X^{k+1-i}
    \end{align}
    where 
    \begin{align}
        g_i(\lambda) 
            & = \sum_{j=0}^{i} b^{jk} \nu_{j}(\lambda)\left\{(-1)^{i-j} b^{\sigma_{i-j}}\begin{bmatrix} k \\ i-j \end{bmatrix}\right\} \\
            & = b^{0}(1)(-1)^{i}b^{\sigma_{i}}\begin{bmatrix} k \\ i \end{bmatrix} + b^{k}(-1)(-1)^{i-1}b^{\sigma_{i-1}}\begin{bmatrix} k \\ i-1 \end{bmatrix}
    \end{align}
    but we have
    \begin{equation}
        \begin{bmatrix} k \\ i \end{bmatrix} \overset{\eqref{equation:gaussianfracx-1k}}{=} \frac{b^{k+1-i}-1}{b^{k+1}-1}\begin{bmatrix} k+1 \\ i \end{bmatrix}
    \end{equation}
    and
    \begin{equation}
        \begin{bmatrix} k \\ i-1 \end{bmatrix} \overset{\eqref{equation:beta1stepdown}}{=} \frac{b^{i}-1}{b^{k+1}-1}\begin{bmatrix} k+1 \\ i \end{bmatrix}.
    \end{equation}
    So, 
    \begin{align}
        g_{i}(\lambda) 
            & = (-1)^{i}b^{\sigma_{i}}\frac{b^{k+1-i}-1}{b^{k+1}-1}\begin{bmatrix} k+1 \\ i \end{bmatrix} + b^{k}(-1)^{i}b^{\sigma_{i-1}}\frac{b^{i}-1}{b^{k+1}-1}\begin{bmatrix} k+1 \\ i \end{bmatrix} 
            \\
            & = (-1)^{i}b^{\sigma_i}\begin{bmatrix} k+1 \\ i \end{bmatrix}\left\{ \frac{b^{k+1-i}-1}{b^{k+1}-1} + b^{k}b^{1-i}\frac{b^{i}-1}{b^{k+1}-1}\right\} 
            \\
            & = (-1)^{i}\frac{b^{\sigma_{i}}\begin{bmatrix} k+1 \\ i \end{bmatrix}}{b^{k+1}-1}\left\{b^{k+1-i} - 1 + b^{k+1} - b^{k+1-i}\right\}
            \\
            & = (-1)^{i}b^{\sigma_{i}}\begin{bmatrix} k+1 \\ i \end{bmatrix}
    \end{align}
    as required. 
\end{proof}
\section{The MacWilliams Identity for the Hermitian Rank Metric}\label{section:MacWilliamsHermitian}

In this section we introduce the Negative-$q$-Krawtchouk polynomials which we then prove are equal to the eigenvalues that are identified in \cite[(4)]{KaiHermitian} of the association scheme of Hermitian matrices over $\mathbb{F}_{q^2}$. In this way a new $q$-analog of the MacWilliams Identity for codes and their duals of Hermitian matrices over $\mathbb{F}_{q^2}$ is presented and proven by comparison with a traditional form of the identity as given in \cite[Theorem 3]{DelsarteAlternating} and proved in \cite{delsarte1973algebraic}.

\subsection{Eigenvalues of the Association Scheme of Hermitian Matrices}

The following definition of the eigenvalues of the association scheme of Hermitian matrices is taken from \cite[(4)]{KaiHermitian} where they have been derived by using the recurrence relation \cite[Lemma 7]{KaiHermitian} together with the original definition of the eigenvalues in \cite[Section V]{InfoTheoryDelsarte} and the formula of the size of each subset of matrices of a given rank from \cite[(1.4)]{CarlitzCountingHermitian}. 
\begin{defn}\label{definition:eigenvaluesofhermitian}
For $x,k\in\left\{ 0,1,\ldots,t\right\}$ with $t\in\mathbb{Z}^+$ the \textbf{\textit{Eigenvalues of the Association Scheme of Hermitian Matrices}}, $Q_k(x,t)$, is defined by 
    \begin{equation}\label{equation:Kaieigenvalues}
        Q_{k}(x,t) = (-1)^k \sum_{j=0}^{k} b^{\sigma_{k-j} +tj} \begin{bmatrix}t-j \\ t-k \end{bmatrix} \begin{bmatrix} t-x \\ j \end{bmatrix}. 
    \end{equation}
\end{defn}
For the rest of this paper we shall call these simply the \textbf{\textit{eigenvalues}}. Note that if we take the definition of the \textbf{\textit{generalised Krawtchouk Polynomials}} from \cite{delsartereccurance} with particular choices of their parameters $b$ and $c$ we then obtain these eigenvalues in Definition \ref{definition:eigenvaluesofhermitian}.

In this paper use is made of the recurrence relation defined in \cite[Lemma 7]{KaiHermitian} which for $x,k\in\left\{0,1,\ldots,t\right\}$ is
\begin{equation}\label{equation:recurrencerelation}
    Q_{k+1}(x+1,t+1) =  Q_{k+1}(x,t+1) + b^{2t+1-x}Q_{k}(x,t).
\end{equation}
It is proven in \cite{KaiHermitian} that $Q_{k}(x,t)$ are the only solutions to the recurrence relation \eqref{equation:recurrencerelation} with initial values
\begin{align}
    Q_0(x,k) & = 1 \label{equation:initalvaluesk=0}\\
    Q_{k}(0,t) & = \xi_{t,k} \label{equation:initalvaluesx=0}\\
        & = \begin{bmatrix} t \\ k \end{bmatrix} \gamma(t,k).
\end{align}


These initial values, $Q_k(0,t)$, count the number of matrices at distance $k$ from any fixed matrix. Now let $\boldsymbol{Q}=\left(q_{xk}\right)$ be the $(t+1)\times (t+1)$ matrix with $q_{xk}=Q_k(x,t)$.
The matrix $\boldsymbol{Q}$ can be used to relate the weight distributions of any code and it's dual. The following theorem is given in \cite{DelsarteAlternating} in relation to alternating bilinear forms but is proved in general for any association scheme in \cite{delsarte1973algebraic}. Here it is written specifically in relation to codes as subgroups of $\mathscr{H}_{q,t}$. It is analogous to the MacWilliams Identity relating the distance distributions of a code and its dual \cite{TheoryofError}\cite{SystematicCode}.

\begin{thm}\label{thm:DelsarteMacWilliams}
Let $\mathscr{C}\subseteq\mathscr{H}_{q,t}$ be a code with weight distribution $\boldsymbol{c}=(c_0,c_1,\ldots,c_t)$ and  $\boldsymbol{\mathscr{C}}^\perp$ be its dual with weight distribution $\boldsymbol{c}'=(c'_0,c'_1,\ldots,c'_t)$. Then,
\begin{equation}\label{equation:delsarte'sMacWills}
    \boldsymbol{c}'=\dfrac{1}{\lvert\mathscr{C}\rvert}\boldsymbol{c}\boldsymbol{P}.
\end{equation}
\end{thm}

\subsection{The Negative-$q$-Krawtchouk Polynomial}

In comparison to the eigenvalues defined above, a new set of polynomials is derived, which arise in finding the negative-$q$-transform $\overline{a}\left(\mu,\nu;t\right)$ where $a(X,Y;t)$ is as defined in Definition \ref{defn:negativeqtransform} and $\mu(X,Y;t)$ and $\nu(X,Y;t)$ are as in Section \ref{subsection:negativeqproduct}. We then go on to prove that these new polynomials are indeed the eigenvalues of the association scheme by proving that they are solutions to the recurrence relation used in \cite{KaiHermitian}.

\begin{defn}
For $t\in\mathbb{Z}^{+},~ x,k\in \{0,1,\ldots,t\}$ we define the \textbf{\textit{Negative-q-Krawtchouk Polynomial}} as
\begin{equation}
    C_{k}(x,t) = \sum_{\ell=0}^{k} (-1)^{\ell}b^{\ell(t-x)}b^{\sigma_{\ell}} \begin{bmatrix} x \\ \ell \end{bmatrix}\begin{bmatrix} t-x \\ k-\ell \end{bmatrix}\gamma(t-\ell, k-\ell).
\end{equation}
\end{defn}

We first prove that the $C_k(x,t)$ satisfy the recurrence relation \eqref{equation:recurrencerelation} and the initial values in Equation \eqref{equation:initalvaluesk=0} and Equation \eqref{equation:initalvaluesx=0} and are therefore the eigenvalues of the association scheme.

\begin{prop}
For all $x,k\in \{0,1,\ldots,t\}$ we have
    \begin{equation}
    C_{k+1}(x+1,t+1) = C_{k+1}(x,t+1) + b^{2t+1-x}C_{k}(x,t).
\end{equation}
\end{prop}

\begin{proof}
   We look at all three terms separately. Firstly,
\begin{align}
    C_{k+1}(x+1,t+1)
        & = \sum_{\ell=0}^{k+1} (-1)^{\ell}b^{\ell(t-x)}b^{\sigma_{\ell}}\begin{bmatrix} x+1 \\ \ell \end{bmatrix}\begin{bmatrix} t-x \\ k+1-\ell \end{bmatrix}\gamma(t+1-\ell,k+1-\ell)
        \\
        & = C_{k+1}(x+1,t+1)\vert_{\ell=k+1} 
        \\
        & \hspace{1cm}\overset{\eqref{equation:gaussiancoeffsx-1k-1}}{+} \sum_{\ell=0}^{k} (-1)^{\ell}b^{\ell(t-x)}b^{\sigma_{\ell}}\left\{\begin{bmatrix} x \\ \ell-1 \end{bmatrix}+b^{\ell}\begin{bmatrix} x \\ \ell \end{bmatrix}\right\}\begin{bmatrix} t-x \\ k+1-\ell \end{bmatrix}\gamma(t+1-\ell,k+1-\ell)
        \\
        & = C_{k+1}(x+1,t+1)\vert_{\ell=k+1} 
        \\
        & \hspace{1cm} + \sum_{\ell=0}^{k} (-1)^{\ell}b^{\ell(t-x)}b^{\sigma_{\ell}}\begin{bmatrix} x \\ \ell-1 \end{bmatrix}\begin{bmatrix} t-x \\ k+1-\ell \end{bmatrix}\gamma(t+1-\ell,k+1-\ell) \label{summand:alpha1}
        \\
        &\hspace{2cm}\overset{\eqref{equation:gammastepdown}}{+} \sum_{\ell=0}^{k}(-1)^{\ell+1}b^{\ell(t-x-1)+k+t+1}b^{\sigma_{\ell}}\begin{bmatrix} x \\ \ell \end{bmatrix}\begin{bmatrix} t-x \\ k+1-\ell \end{bmatrix}\gamma(t-\ell, k-\ell) \label{summand:alpha2}
        \\
        & \hspace{3cm} + \sum_{\ell=0}^{k}(-1)^{\ell+1}b^{\ell(t-x)+k}b^{\sigma_{\ell}}\begin{bmatrix} x \\ \ell \end{bmatrix}\begin{bmatrix} t-x \\ k+1-\ell \end{bmatrix}\gamma(t-\ell,k-\ell) \label{summand:alpha3}
        \\
        & = C_{k+1}(x+1,t+1)\vert_{\ell=k+1} + \alpha_1 + \alpha_2 + \alpha_3
\end{align}
where $\alpha_1, ~ \alpha_2, ~ \alpha_3$ represent \eqref{summand:alpha1}, \eqref{summand:alpha2}, \eqref{summand:alpha3} respectively and for notation $\vert_{j=k+1}$ means ``the term when $j=k+1$". \\
Second,
\begin{align}
    C_{k+1}(x,t+1)
        & = C_{k+1}(x,t+1)\vert_{\ell=k+1} 
        \\
        & \hspace{1cm}\overset{\eqref{equation:gammastepdown}}{+}  \sum_{\ell=0}^{k} (-1)^{\ell+1}b^{l(t-x-1)+t+k+1}b^{\sigma_\ell}\begin{bmatrix} x \\ \ell \end{bmatrix}\begin{bmatrix} t+1-x \\ k-\ell+1 \end{bmatrix}\gamma(t-\ell,k-\ell)
        \\
        & \hspace{2cm} +  \sum_{\ell=0}^{k} (-1)^{\ell+1}b^{l(t-x)+k}b^{\sigma_\ell}\begin{bmatrix} x \\ \ell \end{bmatrix}\begin{bmatrix} t+1-x \\ k-\ell+1 \end{bmatrix}\gamma(t-\ell,k-\ell)
        \\
        & = C_{k+1}(x,t+1)\vert_{\ell=k+1} 
        \\
        &\hspace{1cm}\overset{\eqref{equation:Stepdown1}}{+} \sum_{\ell=0}^{k}(-1)^{\ell+1}b^{\ell(t-x-2)+2k+t+2}b^{\sigma_{\ell}}\begin{bmatrix} x \\ \ell \end{bmatrix}\begin{bmatrix} t-x \\ k-\ell+1 \end{bmatrix}\gamma(t-\ell,k-\ell) \label{summand:beta1}
        \\
        & \hspace{2cm} +\sum_{\ell=0}^{k}(-1)^{\ell+1}b^{\ell(t-x-1)+k+t+1}b^{\sigma_{\ell}}\begin{bmatrix} x \\ \ell \end{bmatrix}\begin{bmatrix} t-x \\ k-\ell \end{bmatrix}\gamma(t-\ell,k-\ell) \label{summand:beta2}
        \\
         & \hspace{3cm}\overset{\eqref{equation:Stepdown1}}{+} \sum_{\ell=0}^{k}(-1)^{\ell+1}b^{\ell(t-x-1)+2k+1}b^{\sigma_{\ell}}\begin{bmatrix} x \\ \ell \end{bmatrix}\begin{bmatrix} t-x \\ k-\ell+1 \end{bmatrix}\gamma(t-\ell,k-\ell) \label{summand:beta3}
        \\
        & \hspace{4cm} + \sum_{\ell=0}^{k}(-1)^{\ell+1}b^{\ell(t-x)+k}b^{\sigma_{\ell}}\begin{bmatrix} x \\ \ell \end{bmatrix}\begin{bmatrix} t-x \\ k-\ell \end{bmatrix}\gamma(t-\ell,k-\ell) \label{summand:beta4}
        \\
        & = C_{k+1}(x,t+1)\vert_{\ell=k+1} + \beta_1 + \beta_2 + \beta_3 + \beta_4
\end{align}
where $\beta_1,~ \beta_2, ~\beta_3, ~\beta_4$ represent \eqref{summand:beta1}, \eqref{summand:beta2}, \eqref{summand:beta3}, \eqref{summand:beta4} respectively.
Thus,
\begin{align}
    b^{2t-x+1}C_{k}(x,t)
        & = \sum_{\ell=0}^{k}(-1)^{\ell}b^{\ell(t-x)+2t-x+1}b^{\sigma_{\ell}}\begin{bmatrix} x \\ \ell \end{bmatrix}\begin{bmatrix} t-x \\ k-\ell \end{bmatrix}\gamma(t-\ell,k-\ell)
        \\
    %
    %
        & = \rho_1
\end{align}
So let $C = C_{k+1}(x+1,t+1) - C_{k+1}(x,t+1) - b^{2t+1-x}C_{k}(x,t)$. so 
\begin{equation}
 C=C_{k+1}(x+1,t+1)\vert_{\ell=k+1}+\alpha_1+\alpha_2+\alpha_3 - C_{k+1}(x,t+1)\vert_{\ell=k+1}- \beta_1 - \beta_2 - \beta_3 - \beta_4 - \rho_1.
\end{equation}
Claim 1: $\alpha_2-\beta_1-\beta_2-\rho_1=0$
\begin{align}
    \alpha_2 - \beta_1 
        & = \sum_{\ell=0}^{k}(-1)^{\ell+1}b^{l(t-x-1)+k+t+1}b^{\sigma_{\ell}}\begin{bmatrix} x \\ \ell \end{bmatrix}\begin{bmatrix} t-x \\ k+1-\ell \end{bmatrix}\gamma(t-\ell,k-\ell)\left\{1-b^{k-\ell+1}\right\}
        \\
        \overset{\eqref{equation:gaussianfracxk-1}}&{=} \sum_{\ell=0}^{k}(-1)^{\ell}b^{\ell(t-x)+t+1}b^{\sigma_{\ell}}\begin{bmatrix} x \\ \ell \end{bmatrix}\begin{bmatrix} t-x \\ k-\ell \end{bmatrix}\gamma(t-\ell,k-\ell)\left(b^{t-x}-b^{k-\ell}\right)
        \\
        & = \rho_1 + \beta_2
\end{align}
Thus $\alpha_2-\beta_1-\beta_2-\rho_1=0$.

Claim 2: $\alpha_1+\alpha_3-\beta_3-\beta_4 = (-1)^k b^{k(t-x+1)+t-x}b^{\sigma_{k+1}}\begin{bmatrix} x \\ k \end{bmatrix}\gamma(t-k,0)$.
\begin{align}
    \alpha_3-\beta_3
        & = \sum_{\ell=0}^{k}(-1)^{\ell+1}b^{\ell(t-x)+k}b^{\sigma_{\ell}}\begin{bmatrix} x \\ \ell \end{bmatrix}\begin{bmatrix} t-x \\ k+1-\ell \end{bmatrix}\gamma(t-\ell,k-\ell)\left\{ 1-b^{k-\ell+1}\right\}
        \\
        \overset{\eqref{equation:gaussianfracxk-1}}&{=} \sum_{\ell=0}^{k}(-1)^{\ell}b^{\ell(t-x+1)}b^{\sigma_{\ell}}\begin{bmatrix} x \\ \ell \end{bmatrix}\begin{bmatrix} t-x \\ k-\ell \end{bmatrix}\gamma(t-\ell,k-\ell)\left(b^{t-x}-b^{k-\ell}\right)
        \\
        & = \sum_{\ell=0}^{k}(-1)^{\ell}b^{\ell(t-x+1)+t-x}b^{\sigma_{\ell}}\begin{bmatrix} x \\ \ell \end{bmatrix}\begin{bmatrix} t-x \\ k-\ell \end{bmatrix}\gamma(t-\ell,k-\ell) + \beta_4 \label{equation:alpha3-beta3-beta4}
\end{align}
So
\begin{align}
    C
        & = C_{k+1}(x+1,t+1)\vert_{\ell=k+1} - C_{k+1}(x,t+1)\vert_{\ell=k+1} +\alpha_1 
        \\
        & \hspace{1cm}+ \sum_{\ell=0}^{k}(-1)^{\ell}b^{\ell(t-x+1)+t-x}b^{\sigma_{\ell}}\begin{bmatrix} x \\ \ell \end{bmatrix}\begin{bmatrix} t-x \\ k-\ell \end{bmatrix}\gamma(t-\ell,k-\ell) \label{equation:rho2}
        \\
        & = C_{k+1}(x+1,t+1)\vert_{\ell=k+1} - C_{k+1}(x,t+1)\vert_{\ell=k+1} +\alpha_1 + \rho_2
\end{align}
where $\rho_2$ represents the summand in \eqref{equation:rho2}.
Now, 
\begin{equation}
    \alpha_1 = \sum_{\ell=1}^{k} (-1)^{\ell}b^{\ell(t-x)}b^{\sigma_{\ell}}\begin{bmatrix} x \\ \ell-1 \end{bmatrix}\begin{bmatrix} t-x \\ k+1-\ell \end{bmatrix}\gamma(t+1-\ell,k+1-\ell)
\end{equation}
since the $\begin{bmatrix} x \\ \ell - 1\end{bmatrix}=0$ when $\ell=0$. Now let $j=\ell-1,~\ell=j+1$
\begin{align}
    \alpha_1
        & = \sum_{j=0}^{k-1} (-1)^{j+1}b^{(j+1)(t-x)}b^{\sigma_{j+1}}\begin{bmatrix} x \\ j \end{bmatrix}\begin{bmatrix} t-x \\ k-j \end{bmatrix}\gamma(t-j,k-j)
        \\
        & = \sum_{\ell=0}^{k-1}(-1)^{\ell+1}b^{\ell(t-x+1)+t-x}b^{\sigma_{\ell}}\begin{bmatrix} x \\ \ell \end{bmatrix}\begin{bmatrix} t-x \\ k-\ell \end{bmatrix}\gamma(t-\ell,k-\ell)
\end{align}
So we have, 
\begin{align}
     \rho_2 + \alpha_1
        & = \sum_{\ell=0}^{k}(-1)^{\ell}b^{\ell(t-x+1)+t-x}b^{\sigma_{\ell}}\begin{bmatrix} x \\ \ell \end{bmatrix}\begin{bmatrix} t-x \\ k-\ell \end{bmatrix}\gamma(t-\ell,k-\ell) 
        \\
        & \hspace{1cm} + \sum_{\ell=0}^{k-1}(-1)^{\ell+1}b^{\ell(t-x+1)+t-x}b^{\sigma_{\ell}}\begin{bmatrix} x \\ \ell \end{bmatrix}\begin{bmatrix} t-x \\ k-\ell \end{bmatrix}\gamma(t-\ell,k-\ell)
        \\
        & = (-1)^k b^{k(t-x+1)+t-x}b^{\sigma_{k+1}}\begin{bmatrix} x\\k\end{bmatrix}\begin{bmatrix} t-x\\0\end{bmatrix}\gamma(t-k,0)
\end{align}
Thus claim 2 holds. So
\begin{align}
    C 
        & = b^{k(t-x)+k+t-x}b^{\sigma_{k}}\begin{bmatrix} x \\ k \end{bmatrix}\begin{bmatrix} t-x \\ 0 \end{bmatrix}\gamma(t-k,0) 
        \\
        & \hspace{1cm} + (-1)^{k+1}b^{(k+1)(t-x)}b^{\sigma_{k+1}}\begin{bmatrix} x+1 \\ k+1 \end{bmatrix}\begin{bmatrix} t-x \\ 0 \end{bmatrix}\gamma(t-k,0)
        \\
        & \hspace{2cm} - (-1)^{k+1}b^{(k+1)(t+1-x)}b^{\sigma_{k+1}}\begin{bmatrix} x \\ k+1 \end{bmatrix}\begin{bmatrix} t+1-x \\ 0 \end{bmatrix}\gamma(t-k,0)
        \\
        & = (-1)^{k}b^{(k+1)(t-x)}b^{\sigma_{k+1}}\left\{\begin{bmatrix} x \\ k \end{bmatrix} - \begin{bmatrix} x+1 \\ k+1 \end{bmatrix} + \begin{bmatrix} x \\ k+1 \end{bmatrix}b^{k+1}\right\}
        \\
        \overset{\eqref{equation:gaussiancoeffsx-1k-1}}&{=} 0.
\end{align}
 
\end{proof}

\begin{lem}\label{lemma:ckiequalsqki}
    The $C_{k}(x,t)$ are the eigenvalues of the association scheme. In other words,
    \begin{equation}\label{equation:ckiequalspki}
        C_{k}(x,t) = Q_{k}(x,t).
    \end{equation}
\end{lem}

\begin{proof}
    The $C_k(x,t)$ satisfy the recurrence relation \eqref{equation:recurrencerelation} and the initial values of the $C_k(x,t)$ are 
\begin{align}
    C_0(x,t)
        & = b^{0}(-1)^{0}b^{0}\begin{bmatrix} x \\ 0 \end{bmatrix}\begin{bmatrix} t-x \\ 0 \end{bmatrix}\gamma(t,0)
        \\
        & = 1
    \\
    C_k(0,t)
        & = \sum_{\ell=0}^{k}b^{\ell t}(-1)^{\ell}b^{\sigma_{\ell}}\begin{bmatrix} 0 \\ \ell \end{bmatrix}\begin{bmatrix} t \\ k-\ell \end{bmatrix}\gamma(t-\ell,k-\ell)
        \\
        & = \begin{bmatrix} t \\ k \end{bmatrix}\gamma(t,k)
        \\
        & = \xi_{t,k}.
\end{align}
\end{proof}

\subsection{The MacWilliams Identity for the Hermitian Rank Metric}\label{section:Macwilliamsidentity}

We now use the Negative-$q$-Krawtchouk polynomials to prove the $q$-analog form of the MacWilliams Identity for hermitian forms over $\mathbb{F}_{q^2}$. This can be seen to be equivalent to what is developed in \cite{KaiHermitian} as a dual inner distribution of a subset is derived using the eigenvalues and the relations of the association scheme. It is then shown that the dual inner distribution defined in this way is the inner distribution of the dual of the considered subset. The new form below uses a functional transformation of the weight distribution rather than explicit use of the eigenvalues. Notably, the form developed in this paper is similar to the $q$-analog of the MacWilliams Identity developed in \cite{gadouleau2008macwilliams} for linear rank metric codes over $\mathbb{F}_{q^m}$ and is similar to the $q$-analog of the MacWilliams Identity developed in \cite{friedlander2023macwilliams} but differs in the form of the Krawtchouk polynomial.

Let the rank weight enumerator of $\mathscr{C}$ be,
\begin{equation}
    W_{\mathscr{C}}^{R}(X,Y)=\sum_{i=0}^t c_i Y^{i} X^{t-i}
\end{equation}
and of it's dual, $\mathscr{C}^{\perp}$ be
\begin{equation}
    W_{\mathscr{C}^\perp}^{R}(X,Y)=\sum_{i=0}^t c_i' Y^{i} X^{t-i}.
\end{equation}

\begin{thm}[The MacWilliams Identity for the Hermitian Rank Metric]\label{mainthm1}
Let $\mathscr{C}$ be a linear code with $\mathscr{C}\subseteq \mathscr{H}_{q,t}$.
Then
\begin{equation}
    W_{\mathscr{C}^\perp}^{R}(X,Y)=\frac{1}{\left| \mathscr{C}\right|}\overline{W}_{\mathscr{C}}^{R}\left( X +(-b^t-1)Y, X-Y\right).
\end{equation}
\end{thm}

\begin{proof}
For $0\leq i\leq n$ we have
\begin{align}
    \left(X-Y\right)^{[i]}\ast\left(X+\left(-b^t-1\right)Y\right)^{[t-i]} 
        & = \left(\sum_{u=0}^i (-1)^u b^{\sigma_{u}}{\begin{bmatrix}i\\u\end{bmatrix}} Y^{u}X^{i-u}\right)\ast\left(\sum_{j=0}^{t-i}{\begin{bmatrix}t-i\\j\end{bmatrix}}\gamma(t,j)Y^{j}X^{t-i-j}\right)
        \\
        \overset{\eqref{equation:qproduct}}&{=}
        \sum_{k=0}^t\left(\sum_{\ell=0}^{k} (-1)^{\ell}b^{\ell(t-i)}b^{\sigma_{\ell}} \begin{bmatrix} i \\ \ell \end{bmatrix}\begin{bmatrix} t-i \\ k-\ell \end{bmatrix}\gamma(t-\ell, k-\ell)\right)Y^{k}X^{t-k}
        \\
        & = \sum_{k=0}^t C_k(i,t)Y^{k}X^{t-k}
        \\
        \overset{\eqref{equation:ckiequalspki}}&{=} \sum_{k=0}^t Q_k(i,t)Y^{k}X^{t-k}.
\end{align}
So then we have
\begin{align}
    \dfrac{1}{\left\vert\mathscr{C} 
    \right\vert} \overline{W}^{R}_\mathscr{C}\left(X+\left(-b^t-1\right)Y, X-Y\right)
        & = \dfrac{1}{\left\vert\mathscr{C}\right\vert}\sum_{i=0}^t c_i\left(X-Y\right)^{[i]}\ast\left(X+\left(-b^t-1\right)Y\right)^{[t-i]}
        \\
        & = \dfrac{1}{\left\vert\mathscr{C}\right\vert}\sum_{i=0}^t c_i \sum_{k=0}^t Q_k(i,t)Y^kX^{t-k}
        \\
        & =\sum_{k=0}^t\left(\dfrac{1}{\left\vert\mathscr{C}\right\vert}\sum_{i=0}^t c_i Q_k(i,t)\right)Y^{k}X^{t-k}
        \\
        \overset{\eqref{equation:delsarte'sMacWills}}&{=} \sum_{k=0}^n c_k' Y^{k}X^{t-k}
        \\
        & = W_{\mathscr{C}^\perp}^{R}(X,Y)
\end{align}
\end{proof}

In this way we have shown that the MacWilliams identity for dual codes based on hermitian forms over $\mathbb{F}_{q^2}$ can be expressed as a $q$-transform of homogeneous polynomials in a form analogous to the original MacWilliams identity for the Hamming metric and the $q$-analogs developed by \cite{gadouleau2008macwilliams} and \cite{friedlander2023macwilliams} for the rank metric and skew rank metric.
\section{The Negative-$q$-Derivatives}\label{section:derivatives}
In this section we develop a new negative-$q$-derivative and negative-$q^{-1}$-derivative to help analyse the coefficients of negative rank weight enumerators. This is analogous to the $q$-derivative applied to the rank metric in \cite{gadouleau2008macwilliams} with the parameter $q$ replaced by $-q=b$.

\subsection{The Negative-$q$-Derivative}

\begin{defn}
For $q \geq 2$, the \textbf{\textit{negative-$\boldsymbol{q}$-derivative}} at $X\neq 0$ for a real-valued function $f(X)$ is defined as
\begin{equation}
    f^{(1)}\left(X\right)=\dfrac{f\left(bX\right)-f\left(X\right)}{(b-1)X}. 
\end{equation}
For $\varphi\geq0$ we denote the $\varphi^{th}$ negative-$q$-derivative (with respect to $X$) of $f(X,Y;\lambda)$ as $f^{(\varphi)}(X,Y;\lambda)$. The $0^{th}$ negative-$q$-derivative of $f(X,Y;\lambda)$ is $f(X,Y;\lambda)$. For any real number $a,~X\neq0,$
\begin{equation}
    \left[ f(X)+ag(X)\right]^{(1)} = f^{(1)}(X)+ag^{(1)}(X).
\end{equation}
\end{defn}


\begin{lem}\label{lemma:negativeqderiv}
~\\
\begin{enumerate}
    \item For $0\leq \varphi \leq \ell,$ $\varphi\in\mathbb{Z}^{+}$ and $\ell \geq 1$,
        \begin{equation}
        \left(X^\ell\right)^{(\varphi)} = \beta(\ell,\varphi)X^{\ell-\varphi}.
        \end{equation}
    \item The $\varphi^{th}$ negative-$q$-derivative of     $f(X,Y;\lambda)=\displaystyle\sum_{i=0}^r f_i(\lambda) Y^{i}X^{r-i}$ is given by
        \begin{equation}\label{equation:vthbderivative}
            f^{(\varphi)}\left(X,Y;\lambda\right)= \displaystyle\sum_{i=0}^{r-\varphi}f_i(\lambda) \beta(r-i,\varphi)Y^{i}X^{r-i-\varphi}.
        \end{equation}
    \item Also,
        \begin{align}
        \mu^{[k](\varphi)}(X,Y;\lambda) & = \beta(k,\varphi)\mu^{[k-\varphi]}(X,Y;\lambda)\label{equation:muderiv}
        \\
        \nu^{[k](\varphi)}(X,Y;\lambda) & = \beta(k,\varphi)\nu^{[k-\varphi]}(X,Y;\lambda).\label{equation:nuderiv}
\end{align}
\end{enumerate}
\end{lem}

\begin{proof}~\\
\begin{enumerate}
    \item[(1)] For $\varphi=1$ we have
        \begin{equation}
            \left(X^{\ell}\right)^{(1)} = \dfrac{\left(bX^{}\right)^\ell-X^{\ell}}{(b-1)X} = \dfrac{b^{\ell}-1}{b-1}X^{\ell-1} = {\begin{bmatrix}\ell\\1\end{bmatrix}}X^{\ell-1} = \beta(\ell,\varphi)X^{\ell-1}.
        \end{equation}
The rest of the proof follows by induction on $\varphi$ and is omitted. 
    \item[(2)] Now consider $f(X,Y;\lambda)=\displaystyle\sum_{i=0}^r f_i (\lambda)Y^{i}X^{r-i}$. We have,
        \begin{align}
            f^{(1)}\left(X,Y;\lambda\right) 
            & = \left(\displaystyle\sum_{i=0}^r f_i (\lambda) Y^{i}X^{r-i}\right)^{(1)}
            \\
            & =\displaystyle\sum_{i=0}^r f_i(\lambda) Y^{i}\left(X^{r-i}\right)^{(1)}
            \\
            & =\displaystyle\sum_{i=0}^{r-1}f_i(\lambda) \beta(r-i,\varphi)Y^{i}X^{r-i-1}
\end{align}

The rest of the proof follows by induction on $\varphi$ and is omitted. 
    \item[(3)] Now consider $\mu^{[k]}=\displaystyle\sum_{u=0}^k \mu_u(\lambda,k)Y^{u}X^{k-u}$ where $\mu_u(\lambda,k) = {\left[\begin{matrix} k \\ u\end{matrix}\right]}\gamma(\lambda,u)$ as in Theorem \ref{equation:muformula}. Then we have
    \begin{align}
        \mu^{[k](1)}(X,Y;\lambda)
            & = \left(\sum_{u=0}^k \mu_u(\lambda,k)Y^{u}X^{k-u}\right)^{(1)}
            \\
            & = \sum_{u=0}^k \mu_u(\lambda,k)Y^{u}\left(\frac{\left(bX\right)^{k-u}-X^{k-u}}{(b-1)X}\right)
            \\
            & = \sum_{u=0}^{k-1}\frac{b^{(k-u)}-1}{b-1}{\begin{bmatrix}k\\u\end{bmatrix}}\gamma(\lambda,u)Y^{u}X^{k-u-1}
            \\
            \overset{\eqref{equation:gaussianfracx-1k}}&{=}  \sum_{u=0}^{k-1}\frac{(b^{k}-1)\left(b^{(k-u)}-1\right)}{(b^{(k-u)}-1)(b-1)}{\begin{bmatrix}k-1\\u\end{bmatrix}}\gamma(\lambda,u)Y^{u}X^{k-u-1}
            \\
            & = \left(\frac{b^{k}-1}{b-1}\right)\mu^{[k-1]}(X,Y;\lambda)
            \\
            \overset{\eqref{equation:betafunction}}&{=}\beta(k,1)\mu^{[k-1]}(X,Y;\lambda)
    \end{align}

So $\mu^{[k](\varphi)}(X,Y;\lambda)=\beta(k,\varphi)\mu^{[k-\varphi]}(X,Y;\lambda)$ follows by induction on $\varphi$ and is omitted.

Now consider $\nu^{[k]}=\displaystyle\sum_{u=0}^k (-1)^u b^{\sigma_{u}}{\begin{bmatrix}k\\u\end{bmatrix}}Y^{u}X^{k-u}$ as in Theorem \ref{lemma:nulemma}. Then we have
    \begin{align}
        \nu^{[k](1)}(X,Y;\lambda) 
            & = \sum_{u=0}^k (-1)^u b^{\sigma_{u}}\frac{b^{(k-u)}-1}{b-1}{\begin{bmatrix}k\\u\end{bmatrix}} Y^{u}X^{k-u-1}
            \\
            & = \sum_{u=0}^{k-1}(-1)^u b^{\sigma_{u}}\frac{\left(b^{k}-1\right)\left(b^{(k-u)}-1\right)}{\left(b^{(k-u)}-1\right)(b-1)}{\begin{bmatrix}k-1\\u\end{bmatrix}}Y^{u}X^{k-1-u}
            \\
            & = \frac{b^{k}-1}{b-1}\nu^{[k-1]}(X,Y;\lambda)
            \\
            & = \beta(k,1)\nu^{[k-1]}(X,Y;\lambda).
    \end{align}
So $\nu^{[k](\varphi)}(X,Y;\lambda)=\beta(k,\varphi)\nu^{[k-\varphi]}(X,Y;\lambda)$ follows by induction also and is omitted. 
\end{enumerate}
\end{proof}

We now need a few smaller lemmas in order to prove Leibniz rule for the negative-$q$-derivative.

\begin{lem}\label{lemma:littleuandv}
Firstly let
    \begin{align}
        u\left(X,Y;\lambda\right) 
            & = \sum_{i=0}^r u_i(\lambda) Y^{i}X^{r-i}
            \\
        v\left(X,Y;\lambda\right) 
            & = \sum_{i=0}^s v_i(\lambda) Y^{i}X^{s-i}.
    \end{align}
\begin{enumerate}
    \item If $u_r=0$ then
        \begin{equation}\label{equation:urequals0}
            \frac{1}{X}\left[u\left(X,Y;\lambda\right)\ast v\left(X,Y;\lambda\right)\right] = \frac{u\left(X,Y;\lambda\right)}{X}\ast v\left(X,Y;\lambda\right).
        \end{equation}
    \item If $v_s=0$ then
        \begin{equation}\label{equation:vsequals0}
            \frac{1}{X}\left[u\left(X,Y;\lambda\right)\ast v\left(X,Y;\lambda\right)\right] = u\left(X, bY;\lambda\right)\ast \frac{v\left(X,Y;\lambda\right)}{X}.
        \end{equation}
\end{enumerate}
\end{lem}
\begin{proof}~\\
    \begin{enumerate}
        \item[(1)] If $u_r=0$,
            \begin{equation}
                \frac{u\left(X,Y;\lambda\right)}{X} = \sum_{i=0}^{r-1}u_i(\lambda) Y^{i}X^{r-i-1}.
            \end{equation}
            Hence
            \begin{align}
                \frac{u\left(X,Y;\lambda\right)}{X}\ast v\left(X,Y;\lambda\right)
                    & = \sum_{k=0}^{r+s-1}\left(\sum_{\ell=0}^k b^{\ell s} u_\ell(\lambda) v_{k-\ell}(\lambda-\ell)\right) Y^{k}X^{r+s-1-k}
                    \\
                    & = \frac{1}{X}\sum_{k=0}^{r+s-1}\left(\sum_{\ell=0}^k b^{\ell s} u_\ell(\lambda) v_{k-\ell}(\lambda-\ell)\right) Y^{k}X^{r+s-k}
                    \\
                    & \hspace{1cm} + \frac{1}{X}\sum_{\ell=0}^{r+s} b^{\ell s} u_\ell(\lambda) v_{r+s-\ell}(\lambda-\ell) Y^{r+s}X^{0}
                    \\
                    & = \frac{1}{X}\left(u\left(X,Y;\lambda\right)\ast v\left(X,Y;\lambda\right)\right)
            \end{align}
        since $v_{r+s-\ell}(\lambda-\ell)=0$ for $0\leq \ell \leq r-1$ and $u_{\ell}(\lambda)=0$ for $r\leq \ell \leq r+s$ so 
        \begin{equation}
            \frac{1}{X}\sum_{\ell=0}^{r+s} b^{\ell s} u_\ell(\lambda) v_{r+s-\ell}(\lambda-\ell) Y^{r+s}X^{0} = 0.
        \end{equation}
    \item[(2)] Now if $v_s=0$,
        \begin{equation}
            \frac{v\left(X,Y;\lambda\right)}{X} = \sum_{i=0}^{s-1} v_i(\lambda) Y^{i}X^{s-1-i}.
        \end{equation}
        Then
        \begin{align}
            u\left(X,bY;\lambda\right) \ast \frac{v\left(X,Y;\lambda\right)}{X} 
                & = \sum_{k=0}^{r+s-1}\left(\sum_{\ell=0}^k b^{\ell(s-1)} b^{\ell}u_\ell(\lambda) v_{k-\ell}(\lambda-\ell)\right)Y^{k}X^{r+s-1-k}
                \\
                & = \sum_{k=0}^{r+s-1}\left(\sum_{\ell=0}^k b^{\ell(s-1)} b^{\ell}u_\ell(\lambda) v_{k-\ell}(\lambda-\ell)\right)Y^{k}X^{r+s-1-k}
                \\
                & \hspace{1cm} + \frac{1}{X}\sum_{\ell=0}^{r+s} b^{\ell s} u_\ell(\lambda) v_{r+s-\ell}(\lambda-\ell) Y^{r+s}X^{0}
                \\
                & = \frac{1}{X}\left[u(X,Y;\lambda)\ast v(X,Y;\lambda)\right]
        \end{align}
        since $v_{r+s-\ell}(\lambda-\ell)=0$ for $0\leq \ell \leq r$ and $u_{\ell}=0$ for $r+1\leq \ell \leq r+s$.
    \end{enumerate}
\end{proof}

\begin{thm}[Leibniz rule for the negative-$q$-derivative] \label{theorem:Leibniznegativeq}
For two homogeneous polynomials in $X$ and $Y$, $f(X,Y;\lambda)$ and $g(X,Y;\lambda)$ with degrees $r$ and $s$ respectively, the $\varphi^{th}$ (for $\varphi\geq0$) negative-$q$-derivative of their negative-$q$-product is given by
    \begin{equation}\label{equation:leibniznegativeq}
        \left[ f\left(X,Y;\lambda\right)\ast g\left(X,Y;\lambda\right)\right]^{(\varphi)} = \sum_{\ell=0}^\varphi {\begin{bmatrix}\varphi \\ \ell \end{bmatrix}}b^{(\varphi-\ell)(r-\ell)}f^{(\ell)}\left(X,Y;\lambda\right)\ast g^{(\varphi-\ell)}\left(X,Y;\lambda\right).
    \end{equation}
\end{thm}

\begin{proof}
Firstly let,
    \begin{align}
        f\left(X,Y;\lambda\right)
            & = \sum_{i=0}^r f_i(\lambda) Y^{i}X^{r-i}
        \\
        u\left(X,Y;\lambda\right) 
            & = \sum_{i=0}^r u_i(\lambda) Y^{i}X^{r-i}
        \\
        g\left(X,Y;\lambda\right) 
            & = \sum_{i=0}^s g_i(\lambda) Y^{i}X^{s-i}
        \\
        v\left(X,Y;\lambda\right) 
            & = \sum_{i=0}^s v_i(\lambda) Y^{i}X^{s-i}.
    \end{align}

For simplification, we shall write $f(X,Y;\lambda)$ as $f(X,Y)$ and similarly for the $g(X,Y;\lambda)$. Now by differentiation we have
\begin{align}
    \left[ f\left(X,Y\right)\ast g\left(X,Y\right)\right]^{(1)} 
        & = \frac{f\left(bX,Y\right)\ast g\left(bX,Y\right)-f\left(X,Y\right)\ast g\left(X,Y\right)}{(b-1)X}
        \\
        & = \frac{1}{(b-1)X} \bigg\{ f\left(bX,Y\right)\ast g\left(bX,Y\right)-f\left(bX,Y\right)\ast g\left(X,Y\right)
        \\
        & \hspace{1cm} + f\left(bX,Y\right)\ast g\left(X,Y\right) - f\left(X,Y\right)\ast g\left(X,Y\right)\bigg\}
        \\
        & = \frac{1}{(b-1)X}\left\{ f\left(bX,Y\right)\ast\left(g\left(bX,Y\right)-g\left(X,Y\right)\right)\right\}
        \\
        & \hspace{1cm} +\frac{1}{(b-1)X}\bigg\{\left(f\left(bX,Y\right)-f\left(X,Y\right)\right)\ast g\left(X,Y\right)\bigg\}
        \\
        \overset{\eqref{equation:vsequals0}}&{=} f\left(bX,bY\right) \ast \left\{\frac{g\left(bX,Y\right)-g\left(X,Y\right)}{(b-1)X}\right\}
        \\
        &\hspace{1cm} \overset{\eqref{equation:urequals0}}{+}\left\{\frac{f\left(bX,Y\right)-f\left(X,Y\right)}{(b-1)X}\right\}\ast g\left(X,Y\right)
        \\
        & = b^{r} f\left(X,Y\right)\ast g^{(1)}\left(X,Y\right) + f^{(1)}\left(X,Y\right)\ast g\left(X,Y\right)
\end{align}

since $f(X,Y)$ is a homogeneous polynomial. So the initial case holds. Assume the statement holds true for $\varphi=\overline{\varphi}$, i.e.

\begin{equation}
    \left[ f\left(X,Y\right)\ast g\left(X,Y\right)\right]^{(\overline{\varphi})} = \sum_{\ell=0}^{\overline{\varphi}} {\begin{bmatrix}\overline{\varphi} \\ \ell \end{bmatrix}}b^{(\overline{\varphi}-\ell)(r-\ell)}f^{(\ell)}\left(X,Y\right)\ast g^{(\overline{\varphi}-\ell)}\left(X,Y\right).
\end{equation}

Now considering $\overline{\varphi}+1$ and for simplicity we write $f(X,Y;\lambda),~g(X,Y;\lambda)$ as $f,g$ we have
\pagebreak
    \begin{align}
    \left[f\ast g\right]^{(\overline{\varphi}+1)} 
        & = \left[ \sum_{\ell=0}^{\overline{\varphi}}{\begin{bmatrix}\overline{\varphi}\\\ell\end{bmatrix}}b^{(\overline{\varphi}-\ell)(r-\ell)}f^{(\ell)}\ast g^{(\overline{\varphi}-\ell)}\right]^{(1)}
        \\
        & = \sum_{\ell=0}^{\overline{\varphi}}{\begin{bmatrix}\overline{\varphi}\\\ell\end{bmatrix}}b^{(\overline{\varphi}-\ell)(r-\ell)}\left[f^{(\ell)}\ast g^{(\overline{\varphi}-\ell)}\right]^{(1)}
        \\
        & = \sum_{\ell=0}^{\overline{\varphi}}{\begin{bmatrix}\overline{\varphi}\\\ell\end{bmatrix}}b^{(\overline{\varphi}-\ell)(r-\ell)}\left( b^{(r-\ell)}f^{(\ell)}\ast g^{(\overline{\varphi}-\ell+1)}+f^{(\ell+1)}\ast g^{(\overline{\varphi}-\ell)}\right)
        \\
    %
    %
    %
        & = \sum_{\ell=0}^{\overline{\varphi}}{\begin{bmatrix}\overline{\varphi}\\\ell\end{bmatrix}}b^{(\overline{\varphi}-\ell+1)(r-\ell)}f^{(\ell)}\ast g^{(\overline{\varphi}-\ell+1)}+ \sum_{\ell=1}^{\overline{\varphi}+1}{\begin{bmatrix}\overline{\varphi}\\\ell-1\end{bmatrix}}b^{(\overline{\varphi}-\ell+1)(r-\ell+1)}f^{(\ell)}\ast g^{(\overline{\varphi}-\ell+1)}
        \\
        & = {\begin{bmatrix}\overline{\varphi}\\0\end{bmatrix}}b^{(\overline{\varphi}+1)r}f\ast g^{(\overline{\varphi}+1)}+ \sum_{\ell=1}^{\overline{\varphi}}{\begin{bmatrix}\overline{\varphi}\\\ell\end{bmatrix}}b^{(\overline{\varphi}+1-\ell)(r-\ell)}f^{(\ell)}\ast g^{(\overline{\varphi}-\ell+1)}
        \\
        & \hspace{1cm} + {\begin{bmatrix}\overline{\varphi}\\\overline{\varphi}\end{bmatrix}}b^{(\overline{\varphi}+1-\overline{\varphi}-1)(r-\overline{\varphi}-1+1)}f^{(\overline{\varphi}+1)}\ast g + \sum_{\ell=1}^{\overline{\varphi}}{\begin{bmatrix}\overline{\varphi}\\\ell-1\end{bmatrix}}b^{(\overline{\varphi}+1-\ell)(r-\ell+1)}f^{(\ell)}\ast g^{(\overline{\varphi}-\ell+1)}
        \\
        & = b^{(\overline{\varphi}+1)r}f\ast g^{(\overline{\varphi}+1)} + f^{(\overline{\varphi}+1)}\ast g+ \sum_{\ell=1}^{\overline{\varphi}} \left( {\begin{bmatrix}\overline{\varphi}\\\ell\end{bmatrix}} + b^{(\overline{\varphi}-\ell+1)}{\begin{bmatrix}\overline{\varphi}\\\ell-1\end{bmatrix}}\right) b^{(\overline{\varphi}-\ell+1)(r-\ell)}f^{(\ell)}\ast g^{(\overline{\varphi}-\ell+1)}
        \\
        \overset{\eqref{equation:Stepdown1}}&{=} \sum_{\ell=1}^{\overline{\varphi}}{\begin{bmatrix}\overline{\varphi}+1\\\ell\end{bmatrix}}b^{(\overline{\varphi}+1-\ell)(r-\ell)}f^{(\ell)}\ast g^{(\overline{\varphi}+1-\ell)}+ {\begin{bmatrix}\overline{\varphi}+1\\0\end{bmatrix}}b^{(\overline{\varphi}+1)(r)}f\ast g^{(\overline{\varphi}+1)}
        \\
        & \hspace{1cm} +{\begin{bmatrix}\overline{\varphi}+1\\\overline{\varphi}+1\end{bmatrix}} b^{(\overline{\varphi}-1-\overline{\varphi}-1)}f^{(\overline{\varphi}+1)}\ast g
        \\
        & = \sum_{\ell=0}^{\overline{\varphi}+1}{\begin{bmatrix}\overline{\varphi}+1\\\ell\end{bmatrix}}b^{(\overline{\varphi}+1-\ell)(r-\ell)}f^{(\ell)}\ast g^{(\overline{\varphi}+1-\ell)}.
    \end{align}
\end{proof}
 
\subsection{The Negative-$q^{-1}$-Derivative}

\begin{defn}
For $q\geq 2, ~(-q)=b$, the \textbf{\textit{negative-$\boldsymbol{q^{-1}}$-derivative}} at $Y\neq 0$ for a real-valued function $g(Y)$ is defined as
    \begin{equation}
        g^{\{1\}}\left(Y\right)=\dfrac{g\left(b^{-1}Y\right)-g\left(Y\right)}{(b^{-1}-1)Y}. 
    \end{equation}
For $\varphi\geq0$ we denote the $\varphi^{th}$ negative-$q^{-1}$-derivative (with respect to $Y$) of $g(X,Y;\lambda)$ as $g^{\{\varphi\}}(X,Y;\lambda)$. The $0^{th}$ negative-$q^{-1}$-derivative of $g(X,Y;\lambda)$ is $g(X,Y;\lambda)$. For any real number $a,~Y\neq0,$
    \begin{equation}
        \left[ f(Y)+ag(Y)\right]^{\{1\}} = f^{\{1\}}(Y)+ag^{\{1\}}(Y).
    \end{equation}
\end{defn}

\begin{lem}\label{lemma:negativeqderivatives} ~ \\
\begin{enumerate}
    \item For $0\leq \varphi \leq \ell,$ 
        \begin{equation}
            \left(Y^\ell\right)^{\{\varphi\}} = b^{\varphi(1-\ell)+\sigma_{\varphi}}\beta(\ell,\varphi)Y^{\ell-\varphi}.
        \end{equation}
    \item The $\varphi^{th}$ negative-$q^{-1}$-derivative of $g(X,Y;\lambda)=\displaystyle\sum_{i=0}^s g_i(\lambda) Y^{i}X^{s-i}$ is given by               \begin{equation}                                    \label{equation:generalbdervispoly}
                g^{\{\varphi\}}\left(X,Y;\lambda\right)=\displaystyle\sum_{i=\varphi}^{s}g_i(\lambda) b^{\varphi(1-i)+\sigma_{\varphi}} \beta(i,\varphi)Y^{i-\varphi}X^{s-i}.
        \end{equation}
    \item Also,
        \begin{align}
            \mu^{[k]\{\varphi\}}(X,Y;\lambda) 
                & = b^{-\sigma_{\varphi}} \beta(k,\varphi)\gamma(\lambda,\varphi)\mu^{[k-\varphi]}(X,Y;\lambda-\varphi)
                \\
            \nu^{[k]\{\varphi\}}(X,Y;\lambda) 
                & = (-1)^{\varphi} \beta(k,\varphi)\nu^{[k-\varphi]}(X,Y;\lambda).
\end{align}
\end{enumerate}
\end{lem}

\begin{proof}~\\
\begin{enumerate}
    \item[(1)] For $\varphi=1$ we have
        \begin{align}
            \left(Y^{\ell}\right)^{\{1\}} = \dfrac{\left(b^{-1}Y\right)^{\ell}-Y^{\ell}}{(b^{-1}-1)Y}
                & = \left(\dfrac{b^{-\ell}-1}{b^{-1}-1}\right)Y^{\ell-1} 
            \\
                & = \frac{bb^{-\ell}\left(1-b^{\ell}\right)}{1-b}Y^{\ell-1}
                \\
                & = b^{1-\ell}\beta(\ell,1)Y^{\ell-1}.
        \end{align}
    So the initial case holds. Assume the case for $\varphi =\overline{\varphi}$ holds. Then we have
        \begin{align}
            \left(Y^{\ell}\right)^{\{\overline{\varphi}+1\}}
                & = \left(b^{(\overline{\varphi}(1-\ell)+\sigma_{\overline{\varphi}})}\beta(\ell,\overline{\varphi})Y^{\ell-\overline{\varphi}}\right)^{\{1\}}
                \\
                & = b^{(\overline{\varphi}(1-\ell)+\sigma_{\overline{\varphi}})}\beta(\ell,\overline{\varphi})\frac{b^{-(\ell-\overline{\varphi})}Y^{\ell-\overline{\varphi}}-Y^{\ell-\overline{\varphi}}}{\left(b^{-1}-1\right)Y}
                \\
    %
    %
            & = b^{\overline{\varphi}(1-\ell)+\sigma_{\overline{\varphi}}}\beta(\ell,\varphi)b^{1-(\ell-\overline{\varphi})}\beta(\ell-\overline{\varphi},1)Y^{\ell-\overline{\varphi}-1}
            \\
            \overset{\eqref{equation:betaproperties}}&{=} b^{(\overline{\varphi}+1)(1-\ell)+\sigma_{\overline{\varphi}+1}}\beta(\ell,\overline{\varphi}+1)Y^{\ell-(\overline{\varphi}+1)}.
        \end{align}
    Thus the statement holds by induction. 
    \item[(2)] Now consider     $g(X,Y;\lambda)=\displaystyle\sum_{i=0}^s g_i (\lambda) Y^{i}X^{s-i}$. For $\varphi=1$ we have
        \begin{equation}
            g^{\{1\}}\left(X,Y;\lambda\right) = \left(\sum_{i=0}^s g_i(\lambda) Y^{i}X^{s-i}\right)^{\{1\}} = \sum_{i=0}^s g_i(\lambda) \left(Y^{i}\right)^{\{1\}}X^{s-i} = \sum_{i=0}^s g_i(\lambda) b^{(-i+1)}\beta(i,1)Y^{i-1}X^{s-i}.
        \end{equation}
    As $\beta(i,1)=0$ when $i=0$ we have
        \begin{equation}
            g^{\{1\}}\left(X,Y;\lambda\right) = \sum_{i=1}^s g_i(\lambda) b^{(1-i)+\sigma_{1}}\beta(i,1)Y^{i-1}X^{s-i}.
        \end{equation}
    So the initial case holds. Now assume the case holds for $\varphi=\overline{\varphi}$ i.e. \\
    $g^{\{\overline{\varphi}\}}\left(X,Y;\lambda\right)=\displaystyle\sum_{i=\overline{\varphi}}^s g_i(\lambda) b^{\overline{\varphi}(1-i)+\sigma_{\overline{\varphi}}}\beta(i,\overline{\varphi})Y^{(i-\overline{\varphi})}X^{s-i}$. Then we have
    \begin{align}
        g^{\{\overline{\varphi}+1\}}\left(X,Y;\lambda\right) 
            & = \left(\sum_{i=\overline{\varphi}}^s g_i(\lambda) b^{\overline{\varphi}(1-i)+\sigma_{\overline{\varphi}}}\beta(i,\overline{\varphi})Y^{i-\overline{\varphi}}\right)^{\{1\}}X^{s-i}
            \\
            & = \sum_{i=\overline{\varphi}}^s g_i(\lambda) b^{\overline{\varphi}(1-i)+\sigma_{\overline{\varphi}}}\beta(i,\overline{\varphi})b^{-(i-\overline{\varphi}-1)}\beta(i-\overline{\varphi},1)Y^{i-\overline{\varphi}-1}X^{s-i}
            \\
            \overset{\eqref{equation:betafunction}}&{=} \sum_{i=\overline{\varphi}}^s g_i(\lambda) b^{(\overline{\varphi}+1)(1-i)+\sigma_{\overline{\varphi}}}\prod_{j=0}^{\overline{\varphi}-1}\frac{\left(b^{i-j}-1\right)\left(b^{i-\overline{\varphi}}-1\right)}{(b-1)(b-1)}Y^{i-\overline{\varphi}-1}X^{s-i}
            \\
            & = \sum_{i=\overline{\varphi}}^s g_i(\lambda) b^{(\overline{\varphi}+1)(1-i)+\sigma_{\overline{\varphi}+1}}\beta(i,\overline{\varphi}+1)Y^{i-\overline{\varphi}-1}X^{s-i}
            \\
            & = \sum_{i=\overline{\varphi}+1}^s g_i(\lambda) b^{(\overline{\varphi}+1)(1-i)+\sigma_{\overline{\varphi}+1}}\beta(i,\overline{\varphi}+1)Y^{i-\overline{\varphi}-1}X^{s-i}
    \end{align}
    since when $i=\overline{\varphi}$, $\beta(\overline{\varphi},\overline{\varphi}+1)=0$. So by induction Equation \eqref{equation:generalbdervispoly} holds.
    
    \item[(3)] Now consider $\mu^{[k]}=\displaystyle\sum_{u=0}^k \mu_u(\lambda,k)Y^{u}X^{k-u}$ where $\mu_u(\lambda,k) = {\left[\begin{matrix} k \\ u\end{matrix}\right]}\gamma(\lambda,u)$ as in Theorem \ref{equation:muformula}. Then we have
    \begin{align}
        \mu^{[k]\{1\}}(X,Y;\lambda) 
            & = \left(\sum_{u=0}^k \mu_u(\lambda,k)Y^{u}X^{k-u}\right)^{\{1\}}
            \\
            & = \sum_{u=0}^k \mu_u(\lambda,k)b^{1-u}\beta(u,1)Y^{u-1}X^{k-u}\\
            & = \sum_{r=0}^{k-1} \mu_{r+1}(\lambda,k)b^{1-(r+1)}\beta(r+1,1)Y^{r+1-1}X^{k-r-1}
            \\
            \overset{\eqref{equation:gammastepdown}\eqref{equation:beta1stepdown}}&{=} \sum_{r=0}^{k-1} {\begin{bmatrix}k\\r+1\end{bmatrix}}\gamma(\lambda,r+1)b^{-r}\beta(r+1,1)Y^{r}X^{k-1-r}
            \\
            & = \sum_{r=0}^{k-1}{\begin{bmatrix}k-1\\r\end{bmatrix}}\frac{b^{k}-1}{b^{(r+1)}-1}\left(b^\lambda-1\right)b^{r}b^{-r}\gamma(\lambda-1,r)\beta(r+1,1)Y^{r}X^{k-1-r}
            \\
            & = b^{-\sigma_1}\beta(k,1)\gamma(\lambda,1)\mu^{[k-1]}(X,Y;\lambda-1).
    \end{align}

    Now assume that the statement holds for $\varphi=\overline{\varphi}$. Then we have
    \begin{align}
        \mu^{[k]\{\overline{\varphi}+1\}}(X,Y;\lambda)
            & = \bigg[ b^{-\sigma_{\overline{\varphi}}}\beta(k,\overline{\varphi})\gamma(\lambda,\overline{\varphi})\mu^{[k-\overline{\varphi}]}(X,Y;\lambda-\overline{\varphi})\bigg]^{\{1\}}
            \\
            & = b^{-\sigma_{\overline{\varphi}}}\beta(k,\overline{\varphi})\gamma(\lambda,\overline{\varphi})\big[\mu^{[k-\overline{\varphi}]}(X,Y;\lambda-\overline{\varphi})\big]^{\{1\}}
            \\
            & = b^{-\sigma_{\overline{\varphi}}}\beta(k,\overline{\varphi})\gamma(\lambda,\overline{\varphi})\left(\sum_{r=0}^{k-\overline{\varphi}}{\begin{bmatrix}k-\overline{\varphi}\\r\end{bmatrix}}\gamma(\lambda-\overline{\varphi},r)Y^{r}X^{k-\overline{\varphi}-r}\right)^{\{1\}}
            \\
            & = b^{-\sigma_{\overline{\varphi}}}\beta(k,\overline{\varphi})\gamma(\lambda,\overline{\varphi})\sum_{r=0}^{k-\overline{\varphi}}{\begin{bmatrix}k-\overline{\varphi}\\r\end{bmatrix}}\gamma(\lambda-\overline{\varphi},r)\left(Y^{r}\right)^{\{1\}}X^{k-\overline{\varphi}-r}
            \\
            & = b^{-\sigma_{\overline{\varphi}}}\beta(k,\overline{\varphi})\gamma(\lambda,\overline{\varphi})\sum_{u=0}^{k-\overline{\varphi}-1}{\begin{bmatrix}k-\overline{\varphi}\\u+1\end{bmatrix}}\gamma(\lambda-\overline{\varphi},u+1)b^{1-(u+1)}\beta(u+1,1)Y^{u+1-1}X^{k-\overline{\varphi}-u-1}
            \\
            \overset{\eqref{equation:gammastepdown}\eqref{equation:betafunction}}&{=} b^{-\sigma_{\overline{\varphi}}}\beta(k,\overline{\varphi})\gamma(\lambda,\overline{\varphi})\sum_{u=0}^{k-(\overline{\varphi}+1)}{\begin{bmatrix}k-\overline{\varphi}-1\\u\end{bmatrix}}\frac{\left(b^{k-\overline{\varphi}}-1\right)\left(b^{u+1}-1\right)}{\left(b^{u+1}-1\right)(b-1)}b^{u}b^{-u}
            \\
            & \hspace{1cm} \times \left(b^{\lambda-\overline{\varphi}}-1\right)\gamma(\lambda-(\overline{\varphi}+1),u)Y^{u}X^{k-(\overline{\varphi}+1)-u}
            \\
            & = b^{-\sigma_{\overline{\varphi}}}b^{-\overline{\varphi}}\gamma(\lambda,\overline{\varphi}+1)\beta(k,\overline{\varphi}+1)\mu^{[k-(\overline{\varphi}+1)]}(X,Y;\lambda-(\overline{\varphi}+1))
            \\
            & = b^{-\sigma_{\overline{\varphi}+1}}\gamma(\lambda,\overline{\varphi}+1)\beta(k,\overline{\varphi}+1)\mu^{[k-(\overline{\varphi}+1)]}(X,Y;\lambda-(\overline{\varphi}+1)).
\end{align}
As required.
Now consider $\nu^{[k]}=\displaystyle\sum_{u=0}^k (-1)^u b^{u(u-1)}{\begin{bmatrix}k\\u\end{bmatrix}}Y^{u}X^{k-u}$ as defined in Theorem \ref{lemma:nulemma}. Then we have
    \begin{align}
        \nu^{[k]\{1\}}(X,Y;\lambda) 
            & = \left(\sum_{u=0}^k (-1)^u b^{\sigma_{u}}{\begin{bmatrix}k\\u\end{bmatrix}}Y^{u}X^{k-u}\right)^{\{1\}}
            \\
            & = \sum_{u=1}^k (-1)^u b^{\sigma_{u}}{\begin{bmatrix}k\\u\end{bmatrix}}\left(Y^{u}\right)^{\{1\}}X^{k-u}
            \\
            & = \sum_{r=0}^{k-1} (-1)^{r+1} b^{\sigma_{r+1}}b^{1-(r+1)}{\begin{bmatrix}k\\r+1\end{bmatrix}}\beta(r+1,1)Y^{r+1-1}X^{k-r-1}
            \\
        \overset{\eqref{equation:gammastepdown}\eqref{equation:betafunction}}&{=} -\sum_{r=0}^{k-1} (-1)^{r} b^{\sigma_{r}}b^{r}b^{-r}{\begin{bmatrix}k-1\\r\end{bmatrix}}\frac{\left(b^{k}-1\right)\left(b^{r+1}-1\right)}{\left(b^{r+1}-1\right)\left(b-1\right)}\beta(r,1)Y^{r}X^{k-r-1}
            \\
            & = (-1)^{1}\beta(k,1)\nu^{[k-1]}(X,Y;\lambda).
\end{align}
Now assume that the statement holds for     $\varphi=\overline{\varphi}$. Then we have
    \begin{align}
        \nu^{[k]}(X,Y;\lambda)^{\{\overline{\varphi}+1\}} 
            & = \left[(-1)^{\overline{\varphi}}\beta(k,\overline{\varphi})\nu^{[k-\overline{\varphi}]}(X,Y;\lambda)\right]^{\{1\}}
            \\
            & = (-1)^{\overline{\varphi}}\beta(k,\overline{\varphi})\sum_{u=1}^{k-\overline{\varphi}} (-1)^u b^{\sigma_u}{\begin{bmatrix}k-\overline{\varphi}\\u\end{bmatrix}}\left(Y^{u}\right)^{\{1\}}X^{k-\overline{\varphi}-u}
            \\
            & = (-1)^{\overline{\varphi}}\beta(k,\overline{\varphi})\sum_{r=0}^{k-\overline{\varphi}-1} (-1)^{r+1} b^{\sigma_{r+1}}b^{-(r+1)+1}{\begin{bmatrix}k-\overline{\varphi}\\r+1\end{bmatrix}}\beta(r+1,1)Y^{r+1-1}X^{k-\overline{\varphi}-r-1}
            \\
            & = (-1)^{\overline{\varphi}+1}\beta(k,\overline{\varphi})\sum_{r=0}^{k-\overline{\varphi}-1} (-1)^{r}b^{\sigma_{r}}{\begin{bmatrix}k-(\overline{\varphi}+1)\\r\end{bmatrix}}
            \times \frac{\left(b^{k-\overline{\varphi}}-1\right)\left(b^{r+1}-1\right)}{\left(b^{r+1}-1\right)\left(b-1\right)}Y^{r}X^{k-\overline{\varphi}-1-r}
            \\
            & = (-1)^{\overline{\varphi}+1}\beta(k,\overline{\varphi}+1)\nu^{[k-(\overline{\varphi}+1)]}(X,Y;\lambda).
\end{align}
as required.
\end{enumerate}
\end{proof}

Now we need a few smaller lemmas in order to prove Leibniz rule for the negative-$q^{-1}$-derivative.

\begin{lem}\label{lemma:negativeqminussimplification} Firstly let
\begin{align}
    u\left(X,Y;\lambda\right) 
        & = \sum_{i=0}^r u_i(\lambda) Y^{i}X^{r-i}
        \\
     v\left(X,Y;\lambda\right) 
        & = \sum_{i=0}^s v_i(\lambda) Y^{i}X^{s-i}.
\end{align}
\begin{enumerate}
\item If $u_0=0$ then
\begin{equation}
    \frac{1}{Y}\left[u\left(X,Y;\lambda\right)\ast v\left(X,Y;\lambda\right)\right] = b^{s}\frac{u\left(X,Y;\lambda\right)}{Y}\ast v\left(X,Y;\lambda-1\right).\label{equation:inverseu=0}
\end{equation}
\item If $v_0=0$ then
\begin{equation}
    \frac{1}{Y}\left[u\left(X,Y;\lambda\right)\ast v\left(X,Y;\lambda\right)\right] = u\left(X, bY;\lambda\right)\ast \frac{v\left(X,Y;\lambda\right)}{Y}.\label{equation:inversev=0}
\end{equation}
\end{enumerate}
\end{lem}

\begin{proof}~\\
\begin{enumerate}
 \item[(1)] Suppose $u_0=0$. Then
\begin{equation}
    \frac{u\left(X,Y;\lambda\right)}{Y} = \sum_{i=0}^{r}u_i(\lambda) Y^{i-1}X^{r-i} = \sum_{i=0}^{r-1} u_{i+1}(\lambda)Y^{i}X^{r-i-1}
\end{equation}
Hence
\begin{align}
    b^{s}\frac{u\left(X,Y;\lambda\right)}{Y}\ast v\left(X,Y;\lambda-1\right) 
        & = b^{s}\sum^{r+s-1}_{u=0}\left(\sum_{\ell=0}^u b^{\ell s}u_{\ell+1}(\lambda)v_{u-\ell}(\lambda-u-1)\right) Y^{u}X^{r+s-1-u}
        \\
        & = b^{s}\sum_{u=0}^{r+s-1}\left(\sum_{i=1}^{u+1}b^{(i-1)s}u_i(\lambda)v_{u-i+1}(\lambda-u-1)\right)Y^{u}X^{r+s-1-u}
        \\
        & = b^{s}\sum_{j=1}^{r+s}\left(\sum_{i=1}^{j}b^{(i-1)s}u_i(\lambda)v_{j-i}(\lambda-j)\right)Y^{j-1}X^{r+s-j}
        \\
        & = \frac{1}{Y}\sum_{j=0}^{r+s}\left(\sum_{i=0}^{j}b^{is}u_i(\lambda)v_{j-i}(\lambda-j)\right)Y^{j}X^{r+s-j}
        \\
        & = \frac{1}{Y}\left(u\left(X,Y;\lambda\right)\ast v\left(X,Y;\lambda\right)\right).
\end{align}
 \item[(2)] Now if $v_0=0$, then
\begin{align}
    \frac{v\left(X,Y;\lambda\right)}{Y} 
        & = \sum_{j=1}^s v_j(\lambda)Y^{j-1}X^{s-j}
        \\
        & = \sum_{i=0}^{s-1} v_{i+1}(\lambda)Y^{i}X^{s-i-1}.
\end{align}
So, 
\begin{align}
    u\left(X,bY;\lambda\right) \ast \frac{v\left(X,Y;\lambda\right)}{Y} 
        & = \sum_{u=0}^{r+s-1}\left(\sum_{j=0}^{u} b^{j(s-1)}b^{j}u_j(\lambda)v_{u-j+1}(\lambda-j)\right)Y^{u}X^{r+s-1-u}
        \\
        & = \sum_{\ell=1}^{r+s}\left(\sum_{j=0}^{\ell-1} b^{js}u_j(\lambda)v_{\ell-j}(\lambda-j)\right)Y^{\ell-1}X^{r+s-\ell}
        \\
        & = \frac{1}{Y}\sum_{\ell=1}^{r+s}\left(\sum_{j=0}^{\ell} b^{js}u_j(\lambda)v_{\ell-j}(\lambda-j)\right)Y^{\ell}X^{r+s-\ell}
        \\
        & = \frac{1}{Y}\sum_{\ell=0}^{r+s}\left(\sum_{j=0}^{\ell} b^{js}u_j(\lambda)v_{\ell-j}(\lambda-j)\right)Y^{\ell}X^{r+s-\ell}
        \\
        & = \frac{1}{Y}\left(u\left(X,Y;\lambda\right)\ast v\left(X,Y;\lambda\right)\right).
\end{align}
\end{enumerate}
\end{proof}

\begin{thm}[Leibniz rule for the negative-$q^{-1}$-derivative]\label{Liebnizbminusderiv}
For two homogeneous polynomials in $Y$, $f(X,Y;\lambda)$ and $g(X,Y;\lambda)$ with degrees $r$ and $s$ respectively, the $\varphi^{th}$ (for $\varphi\geq0$) negative-$q^{-1}$-derivative of their negative-$q$-product is given by

\begin{equation}
    \left[ f\left(X,Y;\lambda\right)\ast g\left(X,Y;\lambda\right)\right]^{\{\varphi\}} = \sum_{\ell=0}^\varphi {\begin{bmatrix}\varphi \\ \ell \end{bmatrix}}b^{\ell(s-\varphi+\ell)}f^{\{\ell\}}\left(X,Y;\lambda\right)\ast g^{\{\varphi-\ell\}}\left(X,Y;\lambda-\ell\right).
\end{equation}
\end{thm}

\begin{proof}
Firstly let,
\begin{align}
    f\left(X,Y;\lambda\right) 
        & = \sum_{i=0}^r f_i(\lambda) Y^{i}X^{r-i}
        \\
    u\left(X,Y;\lambda\right) 
        & = \sum_{i=0}^r u_i(\lambda) Y^{i}X^{r-i}
        \\
    g\left(X,Y;\lambda\right) 
        & = \sum_{i=0}^s g_i(\lambda) Y^{i}X^{s-i}
        \\
    v\left(X,Y;\lambda\right) 
        & = \sum_{i=0}^s v_i(\lambda) Y^{i}X^{s-i}.
\end{align}

For simplification we shall write $f(X,Y;\lambda),~g(X,Y;\lambda)$ as $f(Y;\lambda),~g(Y;\lambda)$. Now by differentiation we have
\begin{align}
    \left[ f\left(Y;\lambda\right)\ast g\left(Y;\lambda\right)\right]^{\{1\}} 
        & = \frac{f\left(b^{-1}Y;\lambda\right)\ast g\left(b^{-1}Y;\lambda\right)-f\left(Y;\lambda\right)\ast g\left(Y;\lambda\right)}{(b^{-1}-1)Y}
        \\
        & = \frac{1}{(b^{-1}-1)Y} \bigg\{ f\left(b^{-1}Y;\lambda\right)\ast g\left(b^{-1}Y;\lambda\right)-f\left(b^{-1}Y;\lambda\right)\ast g\left(Y;\lambda\right)
        \\
        & \hspace{1cm} + f\left(b^{-1}Y;\lambda\right)\ast g\left(Y;\lambda\right) - f\left(Y;\lambda\right)\ast g\left(Y;\lambda\right)\bigg\}
        \\
        & = \frac{1}{(b^{-1}-1)Y}\bigg\{ f\left(b^{-1}Y;\lambda\right)\ast\left(g\left(b^{-1}Y;\lambda\right)-g\left(Y;\lambda\right)\right)\bigg\}
        \\
        & \hspace{1cm} +\frac{1}{(b^{-1}-1)Y}\bigg\{\left(f\left(b^{-1}Y;\lambda\right)-f\left(Y;\lambda\right)\right)\ast g\left(Y;\lambda\right)\bigg\}
        \\
        \overset{\eqref{equation:inversev=0}}&{=} f\left(Y;\lambda\right)\ast\frac{\left(g\left(b^{-1}Y;\lambda\right)-g\left(Y;\lambda\right)\right)}{\left(b^{-1}-1\right)Y}
        \\
        &\hspace{1cm}\overset{\eqref{equation:inverseu=0}}{+}b^{s}\frac{\left(f\left(b^{-1}Y;\lambda\right)-f\left(Y;\lambda\right)\right)}{\left(b^{-1}-1\right)Y}\ast g\left(Y;\lambda-1\right)
        \\
        & = f\left(Y;\lambda\right)\ast g^{\{1\}}\left(Y;\lambda\right) + b^{s} f^{\{1\}}\left(Y;\lambda\right)\ast g\left(Y;\lambda-1\right).
\end{align}
So the initial case holds. Assume the statement holds true for $\varphi=\overline{\varphi}$, i.e.

\begin{equation}
    \left[ f\left(X,Y;\lambda\right)\ast g\left(X,Y;\lambda\right)\right]^{\{\overline{\varphi}\}} = \sum_{\ell=0}^{\overline{\varphi}} {\begin{bmatrix}\overline{\varphi} \\ \ell \end{bmatrix}}b^{\ell(s-\overline{\varphi}+\ell)}f^{\{\ell\}}\left(X,Y;\lambda\right)\ast g^{\{\overline{\varphi}-\ell\}}\left(X,Y;\lambda-2r\right).
\end{equation}
Now considering $\overline{\varphi}+1$  and for simplicity we write $f(X,Y;\lambda),~g(X,Y;\lambda)$ as $f(\lambda),g(\lambda)$ we have

\begin{align}
    \left[ f\left(\lambda\right)\ast g\left(\lambda\right)\right]^{\{\overline{\varphi}+1\}} 
        & = \left[\sum_{\ell=0}^{\overline{\varphi}} {\begin{bmatrix}\overline{\varphi} \\ \ell \end{bmatrix}}b^{\ell(s-\overline{\varphi}+\ell)}f^{\{\ell\}}\left(\lambda\right)\ast g^{\{\overline{\varphi}-\ell\}}\left(\lambda-\ell\right)\right]^{\{1\}}
        \\
        & = \sum_{\ell=0}^{\overline{\varphi}} {\begin{bmatrix}\overline{\varphi} \\ \ell \end{bmatrix}}b^{l(s-\overline{\varphi}+\ell)}\left(f^{\{\ell\}}\left(\lambda\right)\ast g^{\{\overline{\varphi}-\ell\}}\left(\lambda-\ell\right)\right)^{\{1\}}
        \\
        \overset{\text{Def. }\ref{q-proddefn}}&{=} \sum_{\ell=0}^{\overline{\varphi}} {\begin{bmatrix}\overline{\varphi} \\ \ell \end{bmatrix}}b^{\ell(s-\overline{\varphi}+\ell)}f^{\{\ell\}}\left(\lambda\right)\ast g^{\{\overline{\varphi}-\ell+1\}}\left(\lambda-\ell\right)
        \\
        & \hspace{1cm}+ \sum_{\ell=0}^{\overline{\varphi}} {\begin{bmatrix}\overline{\varphi} \\ \ell \end{bmatrix}}b^{\ell(s-\overline{\varphi}+\ell)}b^{(v-\overline{\varphi}+\ell)}f^{\{\ell+1\}}\left(\lambda\right)\ast g^{\{\overline{\varphi}-\ell\}}\left(\lambda-\ell-1\right)
        \\
        & = \sum_{\ell=0}^{\overline{\varphi}} {\begin{bmatrix}\overline{\varphi} \\ \ell \end{bmatrix}}b^{\ell(s-\overline{\varphi}+\ell)}f^{\{\ell\}}\left(\lambda\right)\ast g^{\{\overline{\varphi}-\ell+1\}}\left(\lambda-\ell\right)
        \\
        &  \hspace{2cm}+ \sum_{\ell=1}^{\overline{\varphi}+1} {\begin{bmatrix}\overline{\varphi} \\ \ell-1 \end{bmatrix}}b^{(\ell-1)(s-\overline{\varphi}+\ell-1)}b^{(s-\overline{\varphi}+(\ell-1))}f^{\{\ell\}}\left(\lambda\right)\ast g^{\{\overline{\varphi}-k+1\}}\left(\lambda-\ell\right)
        \\
        & = f\left(\lambda\right)\ast g^{\{\overline{\varphi}+1\}}\left(\lambda\right)+\sum_{\ell=1}^{\overline{\varphi}} {\begin{bmatrix}\overline{\varphi} \\ \ell \end{bmatrix}}b^{\ell(s-\overline{\varphi}+\ell)}f^{\{\ell\}}\left(\lambda\right)\ast g^{\{\overline{\varphi}-\ell+1\}}\left(\lambda-\ell\right)
        \\
        &  \hspace{1cm}+ \sum_{\ell=1}^{\overline{\varphi}} {\begin{bmatrix}\overline{\varphi} \\ \ell-1 \end{bmatrix}}b^{(\ell-1)(s-\overline{\varphi}+\ell-1)}b^{(s-\overline{\varphi}+(\ell-1))}f^{\{\ell\}}\left(\lambda\right)\ast g^{\{\overline{\varphi}-k+1\}}\left(\lambda-\ell\right)
        \\
        &  \hspace{2cm}+ {\begin{bmatrix}\overline{\varphi}\\\overline{\varphi}\end{bmatrix}}b^{(\overline{\varphi}+1)(s+1)}b^{-\overline{\varphi}-1} f^{\{\overline{\varphi}+1\}}\left(\lambda\right)\ast g\left(\lambda-(\overline{\varphi}+1)\right)
        \\
        & = f\left(\lambda\right)\ast g^{\{\overline{\varphi}+1\}}\left(\lambda\right)+ \sum_{\ell=1}^{\overline{\varphi}}\left({\begin{bmatrix}\overline{\varphi}\\\ell\end{bmatrix}}+b^{-\ell}{\begin{bmatrix}\overline{\varphi}\\\ell-1\end{bmatrix}}\right)b^{\ell(s-\overline{\varphi}+\ell)}f^{\{\ell\}}\left(\lambda\right) \ast g^{\{\overline{\varphi}+1-\ell\}}\left(\lambda-\ell\right)
        \\
        &  \hspace{1cm} + b^{s(\overline{\varphi}+1)}b^{-\overline{\varphi}-1}f^{\{\overline{\varphi}+1\}}\left(\lambda\right) \ast g\left(\lambda-(\overline{\varphi}+1)\right)
        \\
        \overset{\eqref{equation:gaussiancoeffsx-1k-1}}&{=} f\left(\lambda\right)\ast g^{\{\overline{\varphi}+1\}}\left(\lambda\right) + \sum_{\ell=1}^{\overline{\varphi}}b^{-\ell}{\begin{bmatrix}\overline{\varphi}+1\\\ell\end{bmatrix}}f^{\{\ell\}}\left(\lambda\right)\ast g^{\{\overline{\varphi}+1\}}\left(\lambda-\ell\right)
        \\
        &  \hspace{1cm} + b^{-(\overline{\varphi}+1)}{\begin{bmatrix}\overline{\varphi}+1\\\overline{\varphi}+1\end{bmatrix}}b^{(\overline{\varphi}+1)(s-\overline{\varphi}(\overline{\varphi}+1))}f^{\{\overline{\varphi}+1\}}\left(\lambda\right)\ast g^{\{\overline{\varphi}+1-(\overline{\varphi}+1)\}}\left(\lambda-(\overline{\varphi}+1)\right)
        \\
        & = \sum_{\ell=0}^{\overline{\varphi}+1}{\begin{bmatrix}\overline{\varphi}+1\\\ell\end{bmatrix}}b^{\ell(s-(\overline{\varphi}+1)+\ell)}f^{\{\ell\}}\left(\lambda\right) \ast g^{\{\overline{\varphi}+1-\ell\}}\left(\lambda-\ell\right)
\end{align}
as required.
\end{proof}

\subsection{Evaluating the Negative-$q$-Derivative and the Negative-$q^{-1}$-Derivative}

The following lemmas yield useful results for applying the MacWilliams Identity to develop moments of the negative rank distribution.

\begin{lem}\label{lemma:d'sandbetab's}
For $X=Y=1$,
\begin{equation}\label{equation:d'snadbetab's}
    \nu^{[j](\ell)}(1,1;\lambda)=\beta(j,j)\delta_{j\ell}.
\end{equation}
\end{lem}

\begin{proof}
Consider 
\begin{align}
    \nu^{[j](\ell)}(X,Y;\lambda) 
        \overset{\eqref{equation:nuderiv}}&{=}  \beta(j,\ell)\nu^{[j-\ell]}(X,Y;\lambda)
        \\
        & = \beta(j,\ell)\sum_{u=0}^{j-\ell}(-1)^ub^{\sigma_u}{\begin{bmatrix}j-\ell\\u\end{bmatrix}}Y^{u}X^{(j-\ell)-u}.
\end{align}
So \begin{equation}
    \nu^{[j](\ell)}(1,1;\lambda) = \beta(j,\ell)\sum_{u=0}^{j-\ell}(-1)^ub^{\sigma_u}{\begin{bmatrix}j-\ell\\u\end{bmatrix}}.
\end{equation}
Now the rest of the proof follows directly from Equation \eqref{equation:deltaijbs}. 

\end{proof}

\begin{lem}\label{lemma:rhoandmu}
For any homogeneous polynomial, $\rho\left(X,Y;\lambda\right)$ and for any $s\geq 0$, \begin{equation}\label{equation:rhoandmu}
\left(\rho \ast \mu^{[s]}\right)\left(1,1;\lambda\right) =(-1)^sb^{\lambda s}\rho(1,1;\lambda).
\end{equation}
\end{lem}

\begin{proof}
Let $\rho\left(X,Y;\lambda\right)=\displaystyle\sum_{i=0}^r \rho_i(\lambda)Y^{i}X^{r-i}$, then from Theorem \ref{equation:muformula}
\begin{equation}
    \mu^{[s]}(X,Y;\lambda) = \sum_{t=0}^s \mu^{[s]}_t(\lambda)Y^{t}X^{s-t} = \sum_{t=0}^{s}{\begin{bmatrix}s\\t\end{bmatrix}}\gamma(\lambda,t)Y^{t}X^{s-t} 
\end{equation}
and
\begin{equation}
    \left(\rho\ast \mu^{[s]}\right)(X,Y;\lambda) = \sum_{u=0}^{r+s}c_u(\lambda)Y^{u}X^{(r+s-u)}
\end{equation}
where
\begin{equation}
    c_u(\lambda)=\sum_{i=0}^u b^{is}\rho_i(\lambda)\mu^{[s]}_{u-i}(\lambda-i).
\end{equation}
Then
\begin{align}
    \left(\rho\ast \mu^{[s]}\right)(1,1;\lambda) 
        & = \sum_{u=0}^{r+s} c_u(\lambda)
        \\
        & = \sum_{u=0}^{r+s}\sum_{i=0}^{u} b^{is}\rho_i(\lambda)\mu^{[s]}_{u-i}(\lambda-i)
        \\
        & = \sum_{j=0}^{r+s}b^{js}\rho_j(\lambda)\left(\sum_{k=0}^{r+s-j}\mu_k^{[s]}(\lambda-j)\right)
        \\
        & = \sum_{j=0}^{r}b^{js}\rho_j(\lambda)\left(\sum_{k=0}^s \mu_k^{[s]}(\lambda-j)\right)
        \\
        & = \sum_{j=0}^r b^{js} \rho_j(\lambda)\left(\sum_{k=0}^s {\begin{bmatrix}s\\k\end{bmatrix}}\gamma(\lambda-j,k)\right)
        \\
        \overset{\eqref{equation:producttosumgauss}}&{=} \sum_{j=0}^r b^{js} \rho_j(\lambda)\left(-b^{\lambda-j}\right)^s
        \\
        & = (-1)^sb^{\lambda s}\rho(1,1;\lambda).
\end{align}
\end{proof}

\section{Moments of the Hermitian Rank Distribution}\label{section:moments}

Here we explore the moments of the Hermitian rank distribution of a subgroup of hermitian forms over $\mathbb{F}_{q^2}$ and that of it's dual. Similar results for the Hamming metric were derived in \cite[p131]{TheoryofError} and for rank metric codes over $\mathbb{F}_{q^m}$ in \cite[Prop 4]{gadouleau2008macwilliams} and for the Skew rank metric is \cite[Section 6]{friedlander2023macwilliams}.

\subsection{Moments derived from the Negative-$q$-Derivative}
The following proposition is obtained in the proof of \cite[Theorem 1]{KaiHermitian} by Kai-Uwe Schmidt, by combining the eigenvalues of the association scheme \cite[(5)]{KaiHermitian} with the entries of the dual inner distribution \cite[(7)]{KaiHermitian}. In this paper an alternative method for deriving the moments is presented using the MacWilliams Identity and the negative-$q$-derivative.
\begin{prop}\label{prop:momentsbderiv}
For $0 \le \varphi \le n, ~(-q)=b,$ and a linear code $\mathscr{C} \subseteq \mathscr{H}_{q,t}$ and its dual $\mathscr{C}^\perp$ with weight distributions ${\boldsymbol{c}}$ and ${\boldsymbol{ c'}}$, respectively we have
\begin{equation}
    \sum_{i=0}^{t-\varphi}{\begin{bmatrix} t-i \\ \varphi\end{bmatrix}}c_i = \frac{1}{|\mathscr{C}^\perp|}\left(-b^{t}\right)^{t-\varphi}\sum_{i=0}^{\varphi} {\begin{bmatrix}t-i \\ t-\varphi\end{bmatrix}}c_i^{'}.
\end{equation}
\end{prop}

\begin{proof}
We apply Theorem \ref{mainthm1} to $\mathscr{C}^\perp$ to get
\begin{equation}
    W_{\mathscr{C}}^{R}(X,Y)=\frac{1}{\left| \mathscr{C}^\perp\right|}\overline{W}_{\mathscr{C}^\perp}^{R}\left( X +(-b^t-1)Y, X-Y\right)
\end{equation}
or equivalently
\begin{align}
    \sum_{i=0}^t c_i Y^{i}X^{t-i}
        & =\frac{1}{\left\vert \mathscr{C}^{\perp}\right|}\sum_{i=0}^t c_i'\left(X-Y\right)^{[i]}\ast \left[X +(-b^t-1)Y\right]^{[t-i]}
        \\
        & = \frac{1}{\left\vert \mathscr{C}^{\perp}\right|} \sum_{i=0}^t c_i'\nu^{[i]}(X,Y;t)\ast \mu^{[t-i]}(X,Y;t). \label{equation:negative-q}
\end{align}

For each side of Equation \eqref{equation:negative-q}, we shall apply the negative-$q$-derivative $\varphi$ times and then evaluate at $X=Y=1$.

For the left hand side, we obtain
\begin{equation}
    \left(\sum_{i=0}^t c_i Y^{i}X^{n-i}\right)^{(\varphi)}=
        \sum_{i=0}^{t-\varphi}c_i \beta(t-i,\varphi)Y^{i}X^{t-i-\varphi}
\end{equation}
from Equation \eqref{equation:vthbderivative}. Putting $X=Y=1$ we then have
\begin{align}
    \sum_{i=0}^{t-\varphi} c_i\beta(t-i,\varphi)
        \overset{\eqref{equation:betabstartdifferent}}&{=} \sum_{i=0}^{t-\varphi}c_i{\begin{bmatrix}t-i \\ \varphi\end{bmatrix}}\beta(\varphi,\varphi)
        \\ 
        & = \beta(\varphi,\varphi)\sum_{i=0}^{t-\varphi}c_i {\begin{bmatrix}t-i \\ \varphi\end{bmatrix}}.
\end{align}

We now move on to the right hand side. For simplicity we write $\mu(X,Y;t)$ as $\mu$ and similarly for $\nu(X,Y;t)$. We also have
\begin{align}
    \left(\frac{1}{\left\vert \mathscr{C}^\perp\right|}\sum_{i=0}^t c_i' \nu^{[i]}\ast \mu^{[t-i]}\right)^{(\varphi)} 
        \overset{\eqref{equation:leibniznegativeq}}&{=} \frac{1}{\left\vert \mathscr{C}^\perp\right|}\sum_{i=0}^t c_i'\left(\sum_{\ell=0}^{\varphi}{\begin{bmatrix}\varphi\\\ell\end{bmatrix}}b^{(\varphi-\ell)(i-\ell)}\nu^{[i](\ell)}\ast \mu^{[t-i](\varphi-\ell)}\right)
        \\
        & = \frac{1}{\left\vert \mathscr{C}^\perp\right|}\sum_{i=0}^t c_i'\psi_i.
\end{align}
Then with $X=Y=1$,
\begin{align}
    \psi_i(X,Y;t) 
        & = \sum_{\ell=0}^{\varphi}{\begin{bmatrix}\varphi\\\ell\end{bmatrix}}b^{(\varphi-\ell)(i-\ell)}\nu^{[i](\ell)}(X,Y;t) \ast \mu^{[t-i](\varphi-\ell)}(X,Y;t)
        \\
    \psi_i(1,1;t) 
        \overset{\eqref{equation:muderiv}}&{=} \sum_{\ell=0}^{\varphi} {\begin{bmatrix}\varphi\\\ell\end{bmatrix}}b^{(\varphi-\ell)(i-\ell)}\beta(t-i,\varphi-\ell)\left(\nu^{[i](\ell)}\ast \mu^{[t-i-\varphi+\ell]}\right)(1,1;t)
        \\
        \overset{\eqref{equation:rhoandmu}}&{=} \sum_{\ell=0}^{\varphi} {\begin{bmatrix}\varphi\\\ell\end{bmatrix}}b^{(\varphi-\ell)(i-\ell)}\beta(t-i,\varphi-\ell)\left(-b^{t}\right)^{t-i-(\varphi-\ell)}\nu^{[i](\ell)}(1,1;t)
        \\
        \overset{\eqref{equation:d'snadbetab's}}&{=} \sum_{\ell=0}^{\varphi}b^{(\varphi-\ell)(i-\ell)}{\begin{bmatrix}\varphi\\\ell\end{bmatrix}}\beta(t-i,\varphi-\ell)\left(-b^{t}\right)^{t-i-(\varphi-\ell)}\beta(i,i)\delta_{i\ell}
        \\
    %
        \overset{\eqref{equation:betabstartdifferent}}&{=} {\begin{bmatrix}\varphi\\i\end{bmatrix}}{\begin{bmatrix}t-i\\ \varphi-i\end{bmatrix}}\beta(\varphi-i,\varphi-i)\left(-b^{t}\right)^{t-\varphi}\beta(i,i)
        \\
        \overset{\eqref{equation:betabstartsame}}&{=} {\begin{bmatrix}t-i \\ \varphi-i\end{bmatrix}}\left(-b^{t}\right)^{t-\varphi}\beta(\varphi,\varphi)
\end{align}

and so

\begin{align}
    \frac{1}{|\mathscr{C}^\perp|}\sum_{i=0}^t c_i'\psi_i(1,1) 
        &  = \frac{1}{|\mathscr{C}^\perp|} \sum_{i=0}^{\varphi}c_i' \begin{bmatrix}t-i\\ \varphi-i\end{bmatrix}\left(-b^{t}\right)^{t-\varphi}\beta(\varphi,\varphi)
        \\
        & = \beta(\varphi,\varphi)\frac{1}{|\mathscr{C}^\perp|}\left(-b^{t}\right)^{t-\varphi}\sum_{i=0}^{\varphi} c_i'{\begin{bmatrix}t-i \\ t-\varphi\end{bmatrix}}.        
\end{align}
Combining the results for each side, and simplifying, we finally obtain
\begin{equation}
    \sum_{i=0}^{t-\varphi}c_i {\begin{bmatrix}t-i\\\varphi\end{bmatrix}} = \frac{1}{|\mathscr{C}^\perp|}\left(-b^{t}\right)^{t-\varphi}\sum_{i=0}^{\varphi} c_i'{\begin{bmatrix}t-i\\t-\varphi\end{bmatrix}}
\end{equation}
as required.
\end{proof}

\begin{note}
In particular, if $\varphi=0$ we have
\begin{equation}
    \sum_{i=0}^t c_i =\frac{\left(-b^{t}\right)^{t}}{|\mathscr{C}^\perp|}c_0' = \frac{\left(-b^{t}\right)^{t}}{|\mathscr{C}^\perp|}.
\end{equation}
In other words
\begin{align}
    |\mathscr{C}||\mathscr{C}^\perp|
        & =\left(-b^{t}\right)^{t}
        \\
        & = (-1)^t (-1)^{t^2} q^{t^2}
        \\
        & = q^{t^2}
\end{align}
as expected.
\end{note}

We can simplify Proposition \ref{prop:momentsbderiv} if $\varphi$ is less than the minimum distance of the dual code.

\begin{cor}\label{corrollary:simplificationpropbderiv}
Let $d_{R}'$ be the minimum rank distance of $\mathscr{C}^\perp$. If $0\leq \varphi < d_{R}'$ then
\begin{equation}
    \sum_{i=0}^{t-\varphi}{\begin{bmatrix} t-i \\ \varphi\end{bmatrix}}c_i = \frac{1}{|\mathscr{C}^\perp|}\left(-b^{t}\right)^{t-\varphi} {\begin{bmatrix}t \\ \varphi\end{bmatrix}}.
\end{equation}
\end{cor}

\begin{proof}
We have $c_0'=1$ and $c_1'=\ldots=c_\varphi'=0$.
\end{proof}


\subsection{Moments derived from the Negative-$q^{-1}$-Derivative}

The next proposition relates the moments of the negative rank distribution of a linear code to those of its dual, this time using the negative-$q^{-1}$-derivative of the MacWilliams Identity for the Hermitian rank association scheme. 
There is a slight difference in the way that these two lemmas are defined compared to the skew rank case (and the rank case presented in \cite[Appendix D]{gadouleau2008macwilliams}). In the skew rank case, $\delta(\lambda,\varphi,j)$ is defined using two gamma functions and a power of $q^2$. This form was effective there because of the particular formulation of the gamma function, which does not hold in this case. Therefore one of the gamma functions in Lemma \ref{lemma:hermitiandeltas} has to be replaced with a more general product. 

\begin{lem}\label{lemma:hermitiandeltas}
Let $\delta(\lambda,\varphi,j)=\displaystyle\sum_{i=0}^{j}\bbinom{j}{i}(-1)^{i}b^{\sigma_{i}}\gamma'(\lambda-i,\varphi)$. Then for all $\lambda\in\mathbb{R},\varphi,j\in\mathbb{Z}$,

\begin{equation}\label{equation:hermitiandelta}
    \delta(\lambda,\varphi,j) = (-1)^{j}\prod_{i=0}^{j-1}\left(b^{\varphi}-b^i\right)\gamma'(\lambda-j,\varphi-j)b^{j(\lambda-j)}.
\end{equation}
\end{lem}
\begin{proof}
We follow the proof by induction. Initial case: $j=0$.
\begin{align*}
    \delta(\lambda,\varphi,0) & = \bbinom{0}{0}(-1)^{0}b^{\sigma_{0}}\gamma'(\lambda,\varphi) = \gamma'(\lambda,\varphi) = \gamma'(\lambda,\varphi)q^{0(\lambda)}.
\end{align*}
So the initial case holds. Now assume it is true for $j=\overline{\jmath}$ and consider the case where $\overline{\jmath}+1$.
\begin{align*}
    \delta(\lambda,\varphi,\overline{\jmath}+1) 
        & = \sum_{i=0}^{\overline{\jmath}+1}\bbinom{\overline{\jmath}+1}{i}(-1)^{i}b^{\sigma_i}\gamma'(\lambda-i,\varphi)
        \\
        \overset{\eqref{equation:gaussiancoeffsx-1k-1}}&{=}\sum_{i=0}^{\overline{\jmath}+1}\left(b^{i}\bbinom{\overline{\jmath}}{i}+\bbinom{\overline{\jmath}}{i-1}\right)(-1)^{i}b^{\sigma_i}\gamma'(\lambda-i,\varphi)
        \\
        & = \sum_{i=0}^{\overline{\jmath}}\bbinom{\overline{\jmath}}{i}(-1)^{i}b^{\sigma_i}b^{i}\gamma'(\lambda-i,\varphi)+\sum_{i=0}^{\overline{\jmath}}\bbinom{\overline{\jmath}}{i}(-1)^{i+1}b^{\sigma_{i+1}}\gamma'(\lambda-(i+1),\varphi)
        \\
        \overset{\eqref{equation:gammastepdown}}&{=}\sum_{i=0}^{\overline{\jmath}}\bbinom{\overline{\jmath}}{i}(-1)^{i}b^{i}b^{\sigma_i}\left(-b^{\lambda -i}-1\right)b^{\varphi-1}\gamma'(\lambda-i-1,\varphi-1)
        \\
        & \hspace{1cm} \overset{\eqref{equation:gammastepdownsecond}}{-}\sum_{i=0}^{\overline{\jmath}}
        \bbinom{\overline{\jmath}}{i}(-1)^{i}b^{\sigma_{i+1}}\left(-b^{\lambda -i-1}-b^{\varphi-1}\right)\gamma'(\lambda-i-1,\varphi-1)
        \\
        & = \sum_{i=0}^{\overline{\jmath}}\bbinom{\overline{\jmath}}{i}(-1)^{i}b^{\sigma_i}\gamma'(\lambda-i-1,\varphi-1)(-b^{\lambda -1})\left(b^{\varphi}-1\right)
        \\
        & = -b^{\lambda -1}\left(b^{\varphi}-1\right)\delta(\lambda-1,\varphi-1,\overline{\jmath})
        \\
        & = -b^{\lambda -1}\left(b^{\varphi}-1\right)(-1)^{\overline{\jmath}}\prod_{i=0}^{{\overline{\jmath}}-1}\left(b^{\varphi-1}-b^i\right)b^{\overline{\jmath}(\lambda-\overline{\jmath}-1)}\gamma'(\lambda-\overline{\jmath}-1,\varphi-\overline{\jmath}-1)
        \\
        \overset{\eqref{equation:gammastepdown}}&{=} (-1)^{\overline{\jmath}+1}b^{(\overline{\jmath}+1)(\lambda-(\overline{\jmath}+1))}\prod_{i=0}^{{\overline{\jmath}}}\left(b^{\varphi}-b^i\right)\gamma'(\lambda-(\overline{\jmath}+1),\varphi-(\overline{\jmath}+1))
\end{align*}
since $\displaystyle\bbinom{\overline{\jmath}}{i-1}=0$ when $i=0$. Hence by induction the lemma is proved.
\end{proof}

\begin{lem}\label{lemma:hermitianepsilons}
Let $\varepsilon(\Lambda,\varphi,i)=\displaystyle\sum_{\ell=0}^{i}\bbinom{i}{\ell}\bbinom{\Lambda-i}{\varphi-\ell}b^{\ell(\Lambda-\varphi)}(-1)^{\ell}b^{\sigma_{\ell}}\prod_{j=0}^{i-\ell-1}\left(b^{\varphi-\ell}-b^j\right)$. Then for all $\Lambda\in\mathbb{R},\varphi,i\in\mathbb{Z}$,
\begin{equation*}
    \varepsilon(\Lambda,\varphi,i) = (-1)^{i}b^{\sigma_{i}}\bbinom{\Lambda-i}{\Lambda-\varphi}.
\end{equation*}
\end{lem}

\begin{proof}
We follow the proof by induction. Initial case $i=0$,
\begin{align*}
    \varepsilon(\Lambda,\varphi,0) = \bbinom{0}{0}\bbinom{\Lambda}{0}b^{0}(-1)^{0}b^{0} & =\bbinom{\Lambda}{\varphi}\\
    (-1)^{0}b^{0}\bbinom{\Lambda}{\Lambda-\varphi} \overset{\eqref{equation:gaussianxx-k}}&{=} 
    \bbinom{\Lambda}{\varphi}.
\end{align*}
So the initial case holds. Now suppose it is true when $i=\overline{\imath}$. Then

\begin{align*}
    \varepsilon(\Lambda,\varphi,\overline{\imath}+1) 
        & = \sum_{\ell=0}^{\overline{\imath}+1}\bbinom{\overline{\imath}+1}{\ell}\bbinom{\Lambda-\overline{\imath}-1}{\varphi-\ell}b^{\ell(\Lambda-\varphi)}(-1)^{\ell}b^{\sigma_{\ell}}\prod_{j=0}^{\overline{\imath}-\ell}\left(b^{\varphi-\ell}-b^j\right)
        \\
        \overset{\eqref{equation:Stepdown1}}&{=} \sum_{\ell=0}^{\overline{\imath}+1}\bbinom{\overline{\imath}}{\ell}\bbinom{\Lambda-\overline{\imath}-1}{\varphi-\ell}b^{\ell(\Lambda-\varphi)}(-1)^{\ell}b^{\sigma_{\ell}}\prod_{j=0}^{\overline{\imath}-\ell}\left(b^{\varphi-\ell}-b^j\right)
        \\
        & \hspace{1cm} + \sum_{\ell=1}^{\overline{\imath}+1}b^{(\overline{\imath}+1-\ell)}\bbinom{\overline{\imath}}{\ell-1}\bbinom{\Lambda-\overline{\imath}-1}{\varphi-\ell}b^{\ell(\Lambda-\varphi)}(-1)^{\ell}b^{\sigma_{\ell}}\prod_{j=0}^{\overline{\imath}-\ell}\left(b^{\varphi-\ell}-b^j\right)
        \\
        & = A + B, \quad \text{say}.
\end{align*}
Now
\begin{align*}
    A 
        & = \left(-b^{\varphi}-b^{\overline{\imath}}\right)\sum_{\ell=0}^{\overline{\imath}}\bbinom{\overline{\imath}}{\ell}\bbinom{\Lambda-\overline{\imath}-1}{\varphi-\ell}b^{\ell(\Lambda-1-\varphi)}(-1)^{\ell}b^{\sigma_{\ell}}\prod_{j=0}^{\overline{\imath}-\ell-1}\left(b^{\varphi-\ell}-b^j\right)
        \\
        & = \left(-b^{\varphi}-b^{\overline{\imath}}\right)\varepsilon(\Lambda-1,\varphi,\overline{\imath})
        \\
        & = \left(-b^{\varphi}-b^{\overline{\imath}}\right)(-1)^{\overline{\imath}}b^{\sigma_{\overline{\imath}}}\bbinom{\Lambda-\overline{\imath}-1}{\Lambda-1-\varphi}
\end{align*}
and
\begin{align*}
    B 
        & =\sum_{\ell=0}^{\overline{\imath}}b^{(\overline{\imath}-\ell)}\bbinom{\overline{\imath}}{\ell}\bbinom{\Lambda-\overline{\imath}-1}{\varphi-\ell-1}b^{(\ell+1)(\Lambda-\varphi)}(-1)^{\ell+1}b^{\sigma_{\ell+1}}\prod_{j=0}^{\overline{\imath}-\ell-1}\left(b^{\varphi-\ell-1}-b^j\right)
        \\
        & = -b^{(\overline{\imath}+\Lambda-\varphi)}\varepsilon(\Lambda-1,\varphi-1,\overline{\imath})
        \\
        & = -b^{(\overline{\imath}+\Lambda                      -\varphi)}(-1)^{\overline{\imath}}b^{\sigma_{\overline{\imath}}}\bbinom{\Lambda-\overline{\imath}-1}{\Lambda-\varphi}.
\end{align*}
So
\begin{align*}
    \varepsilon(\Lambda,\varphi,\overline{\imath}+1) 
        & = A + B 
        \\
        & = (-1)^{\overline{\imath}}b^{\sigma_{\overline{\imath}}}\left\{\left(b^{\varphi}-b^{\overline{\imath}}\right)\bbinom{\Lambda-\overline{\imath}-1}{\Lambda-1-\varphi}-b^{(\overline{\imath}+n-\varphi)}\bbinom{\Lambda-\overline{\imath}-1}{\Lambda-\varphi}\right\}
        \\
        \overset{\eqref{equation:gaussianfracxk-1}}&{=} (-1)^{\overline{\imath}+1}b^{\sigma_{\overline{\imath}}}\left\{b^{\overline{\imath}+n-\varphi}\bbinom{\Lambda-\overline{\imath}-1}{\Lambda-\varphi}-\left(b^{\varphi}-b^{\overline{\imath}}\right)\frac{\left(b^{\Lambda-\varphi}-1\right)}{\left(b^{\varphi-\overline{\imath}}-1\right)}\bbinom{\Lambda-\overline{\imath}-1}{\Lambda-\varphi}\right\}
        \\
        & = (-1)^{\overline{\imath}+1}\bbinom{\Lambda-(\overline{\imath}+1)}{\Lambda-\varphi}b^{\sigma_{\overline{\imath}}}\left\{\frac{b^{\overline{\imath}+n-\varphi}\left(b^{\varphi-\overline{\imath}}-1\right)-\left(b^{\varphi}-b^{\overline{\imath}}\right)\left(b^{\Lambda-\varphi}-1\right)}{\left(b^{\varphi-\overline{\imath}}-1\right)}\right\}
        \\
        & = (-1)^{\overline{\imath}+1}b^{\sigma_{\overline{\imath}+1}}\bbinom{\Lambda-(\overline{\imath}+1)}{\Lambda-\varphi}
\end{align*}
as required.
\end{proof}

\begin{prop}\label{prop:hermitianmomentsbminusderivative}
For $0 \le \varphi \le t$ and a linear code $\mathscr{C} \subseteq \mathscr{H}_{q,t}$ with dimension $k$ and its dual $\mathscr{C}^\perp\subseteq \mathscr{H}_{q,t}$ with weight distributions ${\boldsymbol{c}}=(c_0,\ldots,c_n)$ and ${\boldsymbol{ c'}}=(c'_0,\ldots,c'_n)$, respectively we have
\begin{equation*}
    \sum_{i=\varphi}^t b^{\varphi(t-i)}\bbinom{i}{\varphi}c_i = \frac{1}{|\mathscr{C}^\perp|}\left(-b^{t}\right)^{t-\varphi}\sum_{i=0}^{\varphi}(-1)^{i}b^{\sigma_{i}}b^{i(\varphi-i)}\bbinom{t-i}{t-\varphi}\gamma'(t-i,\varphi-i)c_{i}'.
\end{equation*}
\end{prop}

\begin{proof}
As per Proposition \ref{prop:momentsbderiv}, we apply Theorem \ref{mainthm1} to $\mathscr{C}^{\perp}$ to get

\begin{equation}
    W_{\mathscr{C}}^{R}(X,Y) = \frac{1}{|\mathscr{C}^\perp|}\overline{W}_{\mathscr{C}^{\perp}}^{R}\left(X +(-b^t-1)Y,X-Y\right)
\end{equation}
or equivalently
\begin{align}
    \sum_{i=0}^t c_i Y^{i}X^{t-i} 
        & = \frac{1}{|\mathscr{C}^\perp|}\sum_{i=0}^{t}c_i'\left(X-Y\right)^{[i]}\ast \left(X +(-b^t-1)Y\right)^{[t-i]}
        \\
        & = \frac{1}{|\mathscr{C}^\perp|}\sum_{i=0}^t c_i' \nu^{[i]}(X,Y;t)\ast \mu^{[t-i]}(X,Y;t). \label{equation:negative-q-1}
\end{align}

For each side of Equation \eqref{equation:negative-q-1}, we shall apply the negative-$q^{-1}$-derivative $\varphi$ times and then evaluate at $X=Y=1$. i.e. 

\begin{equation}
    \left(\sum_{i=0}^t c_iY^{i}X^{t-i}\right)^{\{\varphi\}}  = \left(\frac{1}{|\mathscr{C}^\perp|}\sum_{i=0}^t c_i' \nu^{[i]}(X,Y;t)\ast \mu^{[t-i]}(X,Y;t)\right)^{\{\varphi\}}. \label{equation:hermitianmomentswhatwewant}
\end{equation}

For the left hand side, we obtain
\begin{equation}\label{equation:propmomentsbminusLHS}
\begin{split}
    \left(\sum_{i=0}^t c_iY^{i}X^{t-i}\right)^{\{\varphi\}} 
        & = \sum_{i=\varphi}^t c_ib^{\varphi(1-i)+\sigma_{\varphi}}\beta(i,\varphi)Y^{i-\varphi}X^{t-i}
        \\
        \overset{\eqref{equation:betabstartdifferent}}&{=} \sum_{i=\varphi}^{t}c_i b^{\varphi(1-i)+\sigma_{\varphi}} {\begin{bmatrix}i\\\varphi\end{bmatrix}}\beta(\varphi,\varphi)Y^{i-\varphi}X^{t-i}
    \end{split}
\end{equation}
Then using $X=Y=1$ gives
\begin{equation}
    \sum_{i=\varphi}^{t}c_i b^{\varphi(1-i)+\sigma_{\varphi}} {\begin{bmatrix}i\\\varphi\end{bmatrix}}\beta(\varphi,\varphi)Y^{i-\varphi}X^{t-i} = \sum_{i=\varphi}^t b^{\varphi(1-i)+\sigma_\varphi}\beta(\varphi,\varphi){\begin{bmatrix}i\\\varphi\end{bmatrix}} c_i. \label{equation:hermitianLHSmoment}
\end{equation}

We now move on to the right hand side. For simplicity we shall write $\mu(X,Y;t)$ as $\mu(t)$ and similarly $\nu(X,Y;t)$ as $\nu(t)$. Applying Theorem \ref{Liebnizbminusderiv} we get

 \begin{align}\label{equation:propmomentsbminusRHS}
     \left(\frac{1}{|\mathscr{C}^\perp|}\sum_{i=0}^t c_i' \nu^{[i]}(t)\ast \mu^{[t-i]}(t)\right)^{\{\varphi\}} 
        & = \frac{1}{|\mathscr{C}^\perp|}\sum_{i=0}^t c_i' \left(\sum_{\ell=0}^\varphi{\begin{bmatrix}\varphi\\\ell\end{bmatrix}}b^{\ell(t-i-\varphi+\ell)}\nu^{[i]\{\ell\}}(t)\ast \mu^{[t-i]\{\varphi-\ell\}}(t-\ell)\right)
        \\
        & = \frac{1}{|\mathscr{C}^\perp|}\sum_{i=0}^n c_i' \psi_i(t)\label{equation:hermitianRHSpart1moments}
 \end{align}
say. Applying Lemma \ref{lemma:negativeqderivatives} we get
\begin{align}
    \psi_i(t) 
        & = \sum_{\ell=0}^{\varphi} {\begin{bmatrix}\varphi\\l\end{bmatrix}}b^{\ell(t-i-\varphi+\ell)}\left\{(-1)^\ell \beta(i,\ell)\nu^{[i-\ell]}(t)\right\}
        \\
        & \ast\left\{b^{-\sigma_{\varphi-\ell}}\beta(t-i,\varphi-\ell)\gamma(t-\ell,\varphi-\ell)\mu^{[t-i-\varphi+\ell]}(t-\varphi)\right\}.
\end{align}
Now let
\begin{equation}
    \Psi(X,Y;t-\varphi) = \nu^{[i-\ell]}(X,Y;t)\ast\gamma(t - \ell,\varphi-\ell)\mu^{[t-i-\varphi+\ell]}(X,Y;t-\varphi).
\end{equation}
Then we apply the negative-$q$-product and set $X=Y=1$ to get
\begin{align}
    \Psi(1,1;t-\varphi) 
        & = \sum_{u=0}^{t-\varphi}\left[\sum_{p=0}^u b^{p(t-i-\varphi+\ell)}\nu_p^{[i-\ell]}(t)\gamma(t - \ell-p,\varphi-\ell)\mu^{[t-i-\varphi+\ell]}_{u-p}(t-\varphi-p)\right]
        \\
        & = \sum_{r=0}^{i-\ell}b^{r(t-i-\varphi+\ell)}\nu_r^{[i-\ell]}(t)\gamma(t - \ell-r,\varphi-\ell)\left[ \sum_{w=0}^{t-i-\varphi+\ell}\mu_w^{[t-i-\varphi+\ell]}(t-\varphi-r)\right]
        \\
        \overset{\eqref{equation:producttosumgauss}}&{=} \sum_{r=0}^{i-\ell}b^{r(t-i-\varphi+\ell)}(-1)^{t-i-\varphi+\ell}b^{(t-\varphi-r)(t-i-\varphi+\ell)}\nu_r^{[i-\ell]}(t)\gamma(t - \ell-r,\varphi-\ell)
        \\
        & = (-1)^{t-i-\varphi+\ell}b^{(t-\varphi)(t-i-\varphi+\ell)}\sum_{r=0}^{i-\ell}(-1)^r b^{\sigma_{r}}{\begin{bmatrix}i-\ell\\r\end{bmatrix}}\gamma(t - \ell-r,\varphi-\ell)
        \\
        & = (-1)^{t-i-\varphi+\ell}b^{(t-\varphi)(t-i-\varphi+\ell)}(-1)^{i-\ell}b^{(i-\ell)(t-i)}\gamma(\varphi-\ell,i-\ell)\gamma(t-i,\varphi-i)
        \\
        & = (-1)^{t-\varphi}b^{(t-\varphi)(t-i-\varphi+\ell)}b^{(i-\ell)(t-i)}\gamma(\varphi-\ell,i-\ell)\gamma(t-i,\varphi-i)
\end{align}
by Lemma \ref{lemma:hermitiandeltas}. Now using Lemma \ref{lemma:betabmanipulation} and noting that $b^{\ell(t-i-\varphi+\ell)}b^{-\sigma_{\varphi-\ell}}=b^{\ell(t-i)}b^{-\sigma_{\varphi}}b^{\sigma_\ell}$ we get
\begin{align}
    \psi_i(1,1;t) 
        & = \sum_{\ell=0}^{\varphi}(-1)^\ell{\begin{bmatrix}\varphi\\\ell\end{bmatrix}}b^{\ell(t-i-\varphi+\ell)} b^{-\sigma_{\varphi-\ell}}\beta(i,\ell)\beta(t-i,\varphi-\ell)\Psi(1,1;t-\varphi)
        \\
        & = b^{-\sigma_{\varphi}}\beta(\varphi,\varphi)\sum_{\ell=0}^{\varphi}(-1)^{\ell}b^{\ell(t-i)}b^{\sigma_{\ell}}{\begin{bmatrix}i\\\ell\end{bmatrix}}{\begin{bmatrix}t-i\\ \varphi-\ell\end{bmatrix}}\Psi(1,1;t-\varphi).
\end{align}
Writing that
\begin{align}
    b^{-\sigma_{\varphi}}b^{\ell(t-i)}b^{(t-\varphi)(t-\varphi-i+\ell)}b^{(i-\ell)(t-i)} 
        & = b^{\sigma_{\varphi}}b^{\varphi(1-t)}b^{t(t-\varphi)}b^{\ell(t-\varphi)}b^{i(\varphi-i)}
        \\
        & = b^{\theta}b^{\ell(t-\varphi)}
\end{align}
we get
\begin{align}
    \psi_i(1,1;t) 
        & = (-1)^{t-\varphi}b^{\theta}\beta(\varphi,\varphi)\gamma(t-i,\varphi-i)\sum_{\ell=0}^i(-1)^\ell b^{\ell(t-\varphi)} b^{\sigma_{\ell}}{\begin{bmatrix}i\\\ell\end{bmatrix}}{\begin{bmatrix}t-i\\\varphi-\ell\end{bmatrix}}\gamma(\varphi-\ell,i-\ell) 
        \\
        & = (-1)^{t-\varphi}(-1)^ib^{\theta}b^{\sigma_{i}}\beta(\varphi,\varphi){\begin{bmatrix}t-i\\t-\varphi\end{bmatrix}}\gamma(t-i,\varphi-i)\label{equation:hermitianmomentsRHSpsi}
\end{align}
by Lemma \ref{lemma:hermitianepsilons}.

Substituting the results from \eqref{equation:hermitianLHSmoment}, \eqref{equation:hermitianRHSpart1moments} and \eqref{equation:hermitianmomentsRHSpsi} we have
\begin{equation}
    \sum_{i=\varphi}^t b^{\varphi(1-i)+\sigma_\varphi}\beta(\varphi,\varphi){\begin{bmatrix}i\\\varphi\end{bmatrix}} c_i = \frac{1}{|\mathscr{C}^\perp|}\sum_{i=0}^t c_i'(-1)^{t-\varphi+i}b^{\theta}b^{\sigma_{i}}\beta(\varphi,\varphi){\begin{bmatrix}t-i\\t-\varphi\end{bmatrix}}\gamma(t-i,\varphi-i).
\end{equation}
Thus cancelling and rearranging gives,
\begin{equation}
    \sum_{i=\varphi}^t b^{\varphi(t-i)} {\begin{bmatrix}i \\ \varphi\end{bmatrix}}c_i  = \frac{\left(-b^{t}\right)^{t-\varphi}}{|\mathscr{C}^\perp|}\sum_{i=0}^\varphi  (-1)^{i} b^{\sigma_i}b^{i(\varphi-i)}{\begin{bmatrix}t-i \\ t-\varphi\end{bmatrix}}\gamma(t-i,\varphi-i)c_i'
\end{equation}
as required.
\end{proof}

We can simplify Proposition \ref{prop:hermitianmomentsbminusderivative} if $\varphi$ is less than the minimum distance of the dual code. Also we can introduce the \textbf{\textit{diameter}}, $\varrho_{R}'$, to be the maximum distance between any two codewords of the dual code and simplify Proposition \ref{prop:hermitianmomentsbminusderivative} again.

\begin{cor}
If $0\leq \varphi < d_{R}'$ then
\begin{equation}
    \sum_{i=\varphi}^t b^{\varphi(t-i)}{\begin{bmatrix} i \\ \varphi\end{bmatrix}}c_i 
        = 
    \frac{1}{|\mathscr{C}^\perp|}\left(-b^{t}\right)^{t-\varphi}{\begin{bmatrix} t \\ \varphi\end{bmatrix}}\gamma(t,\varphi).
\end{equation}
For $\varrho_{R}'< \varphi \leq t$ then
\begin{equation}
   \sum_{i=0}^{\varphi}(-1)^{i}{b^{\sigma_{i}}b^{i(\varphi-i)}\begin{bmatrix} t-i \\ t-\varphi\end{bmatrix}}\gamma(t-i,\varphi-i)c_{i}
        =
    0.
\end{equation}
\end{cor}

\begin{proof}
First consider $0\leq\varphi< d_{R}'$, then $c'_0=1$, $c_1'=\ldots=c_\varphi'=0$. Also since ${\begin{bmatrix} t \\ t-\varphi
\end{bmatrix}}={\begin{bmatrix} t \\ \varphi
\end{bmatrix}}$ the statement holds. Now if $\varrho_{R}'<\varphi\leq n$ then applying Proposition \ref{prop:hermitianmomentsbminusderivative} to $\mathscr{C}^\perp$ gives
\begin{equation}
    \sum_{i=\varphi}^t b^{\varphi(t-i)}{\begin{bmatrix}i\\\varphi\end{bmatrix}}c'_i 
        = 
    \frac{1}{|\mathscr{C}|}\left(-b^{t}\right)^{t-\varphi}\sum_{i=0}^{\varphi}(-1)^{i}{b^{\sigma_{i}}b^{i(\varphi-i)}\begin{bmatrix}t-i\\ t-\varphi\end{bmatrix}}\gamma(t-i,\varphi-i)c_{i}.
\end{equation}
So using $c_\varphi'=\ldots=c_t'=0$ we get
\begin{equation}
    0 = \sum_{i=0}^{\varphi}(-1)^{i}{b^{\sigma_{i}}b^{i(\varphi-i)}\begin{bmatrix}t-i\\ t-\varphi\end{bmatrix}}\gamma(t-i,\varphi-i)c_{i}
\end{equation}
as required.
\end{proof}

\subsection{MHRD Codes}\label{section:MHRD}

As an application for the MacWilliams Identity, we can derive an alternative proof for the explicit coefficients of the Hermitian rank weight distribution for some MHRD codes to that in \cite[Theorem 3]{KaiHermitian}. This is analogous to the results for MRD codes presented in \cite[Proposition 9]{gadouleau2008macwilliams} and \cite[Proposition 6.8]{friedlander2023macwilliams}. Firstly a lemma, analogous to the rank and skew rank cases, that will be needed.
\begin{lem}\label{lemma:hermitiansequences}
If $a_0,a_1,\ldots,a_\ell$ and $b_0,b_1,\ldots,b_\ell$ are two sequences of real numbers and if
\begin{equation*}
    a_j = \sum_{i=0}^{j}\bbinom{\ell-i}{\ell-j}b_i
\end{equation*}
for $0\leq j\leq \ell$, then also for $0 \leq i \leq \ell$ we have,
\begin{equation*}
    b_i = \sum_{j=0}^i (-1)^{i-j}b^{\sigma_{i-j}}\bbinom{\ell-j}{\ell-i}a_j.
\end{equation*}
\end{lem}
\begin{proof}
For $0\leq i \leq \ell$, 
\begin{align*}
    \sum_{j=0}^i (-1)^{i-j}b^{\sigma_{i-j}}\bbinom{\ell-j}{\ell-i}a_j 
        & = \sum_{j=0}^i(-1)^{i-j}b^{\sigma_{i-j}}\bbinom{\ell-j}{\ell-i}\left(\sum_{k=0}^j \bbinom{\ell-k}{\ell-j}b_k\right)
        \\
        & = \sum_{k=0}^i \sum_{j=k}^i (-1)^{i-j}b^{\sigma_{i-j}}\bbinom{\ell-j}{\ell-i}\bbinom{\ell-k}{\ell-j}b_k
        \\
        & = \sum_{k=0}^i b_k \left(\sum_{s=\ell-i}^{\ell-k}(-1)^{i-\ell+s}b^{\sigma_{i-\ell+s}}\bbinom{s}{\ell-i}\bbinom{\ell-k}{s}\right)
        \\
        \overset{\eqref{equation:deltaijbs}}&{=} \sum_{k=0}^i b_k \delta_{ik}
        \\
        & = b_i
\end{align*}
as required.
\end{proof}
Before we write our next proposition, we shall explain a similar proposition presented by Schmidt \cite[Theorem 3]{KaiHermitian}. Theorem 3 states that if a code, $\mathscr{C}$, has minimum distance $d_{HR}$, and its dual, $\mathscr{C}^\perp$, has minimum distance at least $t-d_{HR}+1$, then the weight distribution is uniquely determined by its parameters. Moreoever, if $d_{HR}$ is odd and the code $\mathscr{C}$ meets the Singleton bound, i.e. $\vert\mathscr{C}\vert=q^{t(t-d_{HR}+1)}$ \eqref{SinglebuondHermitian} then, by \cite[Theorem 1]{KaiHermitian}, $\mathscr{C}^\perp$ has minimum distance at least $t-d_{HR}+2$ and the conditions for \cite{KaiHermitian} Theorem 3 are met. However, if $d_{HR}$ is even and $\mathscr{C}$ meets the Singleton bound, the weight distribution cannot necessarily be determined uniquely by its parameters and Schmidt provides a counterexample to show this. Consequently the following proposition looks specifically at codes which are MHRD, i.e. meets the Singleton bound, with minimum distance, $d_{HR}$, odd. We can then use \cite[Theorem 1]{KaiHermitian} and Corollary \ref{corrollary:simplificationpropbderiv} to derive the unique weight distribution of the code as a function of its parameters equivalent to \cite[Theorem 3]{KaiHermitian}.

\begin{prop}
    Let $\mathscr{C}\subseteq\mathscr{H}_{q,t}$ be a linear MHRD code with weight distribution $\boldsymbol{c}=(c_0,\ldots,c_t)$ and minimum distance $d_{HR}$ odd. Then we have $c_0=1$ and for $0\leq r \leq t-d_{HR}$,
\begin{equation*}
    c_{r+d_{HR}} = \sum_{i=0}^r (-1)^{r-i}b^{\sigma_{r-i}}\bbinom{d_{HR}+r}{d_{HR}+i}\bbinom{t}{d_{HR}+r}\left(\frac{\left(-b^{t}\right)^{d_{HR}+i}}{|\mathscr{C}^\perp|}-1\right).
\end{equation*}
\end{prop}

\begin{proof}
From Corollary \ref{corrollary:simplificationpropbderiv}, for $0 \leq \varphi < d_{HR}'$ we have
\begin{equation*}
    \sum_{i=0}^{t-\varphi}\bbinom{t-i}{\varphi}c_i = \frac{1}{|\mathscr{C}^\perp|}\left(-b^{t}\right)^{t-\varphi} \bbinom{t}{\varphi}.
\end{equation*}
    Now if a linear code $\mathscr{C}$ is MHRD, with minimum distance $d_{HR}$ odd, then $\mathscr{C}^\perp$ is also MHRD with minimum distance $d_{HR}'=t-d_{HR}+2$ by \cite[Theorem 1]{KaiHermitian}. So Corollary \ref{corrollary:simplificationpropbderiv} holds for $0\leq \varphi \leq t-d_{HR}=d_{HR}'-2$. We therefore have $c_0=1$ and $c_1=c_2=\ldots=c_{d_{HR}-1}=0$ and setting $\varphi=t-d_{HR}-j$ for $0\leq j\leq t-d_{HR}$ gives
    \begin{align*}
    \bbinom{t}{t-d_{HR}-j} + \sum_{i=d_{HR}}^{d_{HR}+j}\bbinom{t-i}{t-d_{HR}-j}c_i 
        & = \frac{1}{|\mathscr{C}^\perp|}\left(-b^{t}\right)^{d_{HR}+j} \bbinom{t}{t-d_{HR}-j}
        \\
    \sum_{r=0}^j \bbinom{t-d_{HR}-r}{t-d_{HR}-j}c_{r+d_{HR}} 
        & = \bbinom{t}{t-d_{HR}-j}\left(\frac{\left(-b^{t}\right)^{d_{HR}+j}}{|\mathscr{C}^\perp|}-1\right).
\end{align*}
Applying Lemma \ref{lemma:hermitiansequences} with $\ell = t-d_{HR}$ and $b_r = c_{r+d_{HR}}$ then setting 
\begin{equation*}
    a_j = {\bbinom{t}{t-d_{HR}-j}}\left(\frac{\left(-b^{t}\right)^{d_{HR}+j}}{|\mathscr{C}^\perp|}-1\right)
\end{equation*}
gives
\begin{equation*}
     \sum_{r=0}^j \bbinom{t-d_{HR}-r}{t-d_{HR}-j}b_r = a_j
\end{equation*}
and so
\begin{align*}
    b_r = c_{r+d_{HR}} 
        & = \sum_{i=0}^r (-1)^{r-i}b^{\sigma_{r-i}}\bbinom{t-d_{HR}-i}{t-d_{HR}-r}a_i
        \\
        & = \sum_{i=0}^r (-1)^{r-i}b^{\sigma_{r-i}}\bbinom{t-d_{HR}-i}{t-d_{HR}-r}\bbinom{t}{t-d_{HR}-i}\left(\frac{\left(-b^{t}\right)^{d_{HR}+i}}{|\mathscr{C}^\perp|}-1\right).
\end{align*}
But we have
\begin{align*}
    \bbinom{t-d_{HR}-i}{t-d_{HR}-r}\bbinom{t}{t-d_{HR}-i} 
        \overset{\eqref{equation:gaussianxx-k}}&{=} \bbinom{t-(d_{HR}+i)}{t-(d_{HR}+r)}\bbinom{t}{d_{HR}+i}
        \\
    \overset{\eqref{equation:gaussianswapplaces}}&{=}    
        \bbinom{d_{HR}+r}{d_{HR}+i}\bbinom{t}{t-(d_{HR}+r)}
        \\
    \overset{\eqref{equation:gaussianxx-k}}&{=}
        \bbinom{d_{HR}+r}{d_{HR}+i}\bbinom{t}{d_{HR}+r}.
\end{align*}
Therefore
\begin{equation*}
    c_{r+d_{HR}} = \sum_{i=0}^r (-1)^{r-i}b^{\sigma_{r-i}}\bbinom{d_{HR}+r}{d_{HR}+i}\bbinom{t}{d_{HR}+r}\left(\frac{\left(-b^{t}\right)^{d_{HR}+i}}{|\mathscr{C}^\perp|}-1\right)
\end{equation*}
as required.
\end{proof}

\newpage
\printbibliography

\end{document}